\documentclass[11pt]{article}
\usepackage[margin=1in]{geometry}
\usepackage{natbib}
\usepackage{booktabs} 
\usepackage[ruled]{algorithm2e} 

\SetAlFnt{\small}
\SetAlCapFnt{\small}
\SetAlCapNameFnt{\small}
\SetAlCapHSkip{0pt}
\IncMargin{-\parindent}
\usepackage{xcolor}

\usepackage{booktabs}
\usepackage{bm}
\usepackage{nicematrix}
\usepackage{float}
\usepackage{dsfont}
\usepackage{tikz}
\usetikzlibrary{positioning}

\usepackage{amsthm}

\makeatletter
\def\th@plain{%
  \thm@notefont{}%
  \itshape
  \thm@headfont{\scshape}%
}
\makeatother

\usepackage{amssymb}
\usepackage{mathrsfs}
\usepackage{color,soul}
\usepackage{mathtools}

\usepackage{multirow}
\usepackage{rotating}
\usepackage{makecell}
\usepackage{caption}

\usepackage{hyperref}
\usepackage{thmtools}
\usepackage{thm-restate}
\usepackage{enumitem}

\usepackage{cleveref}
\usepackage{graphicx} 

\usepackage[]{color-edits}
\addauthor[Ram]{rd}{blue}
\addauthor[Sid]{sb}{teal}
\addauthor[Bobby]{rk}{magenta}
\addauthor[TODO]{todo}{red}

\newtheorem{lemma}{Lemma}
\newtheorem*{lemma*}{Lemma}
\newtheorem{observation}{Observation}
\newtheorem{definition}{Definition}

\newtheorem{fact}{Fact}
\newtheorem{proposition}{Proposition}


\crefname{protocol}{protocol}{protocols}
\Crefname{protocol}{Protocol}{Protocols}
\crefname{observation}{observation}{observations}
\Crefname{observation}{Observation}{Observations}
\crefname{enumi}{property}{properties}
\Crefname{enumi}{Property}{Properties}

\newcommand{\mbf}[1]{\mathbf{#1}}

\newcommand{\opt}{\textsc{OPT}}

\newcommand{\uniform}{\textsc{Uniform}}

\newcommand{\wf}{\textsc{WF}}

\newcommand{\Enest}{\mathcal{E}^{\textsc{nest}}}

\newcommand{\ratio}{\textsc{CR}}
\newcommand{\reg}{\textsc{Reg}}

\newcommand{\waterfill}{\textsc{water-filling}\xspace}

\newcommand{\nsw}{\textsc{NSW}}
\newcommand{\fm}{\textsc{FM}}

\newcommand{\1}[1]{\mathds{1}\left\{#1\right\}}

\newcommand{\pp}[1]{\left( #1 \right)^+}

\newcommand{\wh}[1]{\widehat{#1}}
\newcommand{\wt}[1]{\widetilde{#1}}

\newcommand{\E}{\mathbb{E}}
\newcommand{\N}{\mathbb{N}}
\newcommand{\R}{\mathbb{R}}
\newcommand{\calA}{\mathcal{A}}

\newcommand{\calE}{\mathcal{E}}

\newcommand{\calI}{\mathcal{I}}

\newcommand{\bell}{\bm{\ell}}
\newcommand{\wtE}{\widetilde{E}}
\newcommand{\define}{\stackrel{\Delta}{=}}

\usepackage[scr=boondox]{mathalfa}

\setcitestyle{authoryear}

\newcommand{\acknowledgements}{\thanks{
The authors were supported in part by AFOSR grant FA9550-23-1-0068, NSF grants ECCS-1847393 and CNS-195599, and a Cornell Duffield Engineering SPROUT grant.
}}

\title{Water-Filling is Universally Minimax Optimal\acknowledgements}

\author{
    Siddhartha Banerjee
    \thanks{
        Cornell University.
        Emails: {\tt sbanerjee@cornell.edu, ramdcv@cs.cornell.edu, rdk@cs.cornell.edu}
    }
    \and Ramiro N. Deo-Campo Vuong\footnotemark[2]
    \and Robert Kleinberg\footnotemark[2]
}
\date{}

\begin{document}
\maketitle

\begin{abstract}
    Allocation of dynamically-arriving (i.e.,  online) divisible resources among a set of offline agents is a fundamental problem, with applications to online marketplaces, scheduling, portfolio selection, signal processing, and many other areas.
    The \waterfill algorithm, which allocates an incoming resource to maximize the minimum load of compatible agents, is ubiquitous in many of these applications whenever the underlying objectives prefer more balanced solutions; however, the analysis and guarantees differ across settings. 

    We provide a justification for the widespread use of \waterfill by showing that it is a \emph{universally minimax optimal} policy in a strong sense. 
    Formally, our main result implies that \waterfill is minimax optimal for a large class of objectives -- including both Schur-concave maximization and Schur-convex minimization -- under $\alpha$-regret and competitive ratio measures.
    This optimality holds for every fixed tuple of agents and resource counts.
    Remarkably, \waterfill achieves these guarantees as a myopic policy, remaining entirely agnostic to the objective function, agent count, and resource availability.    

    Our techniques notably depart from the popular primal-dual analysis of online algorithms, and instead develop a novel way to apply the theory of majorization in online settings to achieve universality guarantees.
\end{abstract}

\section{Introduction}
\label{sec:intro}

Allocation of divisible resources (goods or chores) to agents is a fundamental algorithmic primitive. In many applications, the resources must be assigned dynamically as they arrive, under uncertainty about the future availability of resources and their compatibility with different agents, and in a way that is equitable. For example, a funding organization may allocate grants to eligible recipients over a sequence of time periods, without knowing the amount of funds that will be available for distribution in future periods nor which sets of recipients will be eligible. In any such setting, one may ask which allocation policies are minimax optimal, in the sense that their worst-case outcome is as equitable as possible.
It would appear that this goal is underspecified; surely the answer must depend on the metric used to evaluate equitable allocations. The main message of this paper is that, despite the question's ambiguity, it has an essentially unambiguous answer: the well-known and widely used \waterfill algorithm, which greedily allocates arriving resources to maximize the minimum load among eligible recipients.

To this point, the ubiquity of \waterfill --- also known as max-min fairness or resource balancing --- has largely been explained through case-by-case analysis of different equity-promoting objectives. In each of the cases considered, the \waterfill algorithm was shown to attain the best possible competitive ratio guarantee in the limit as the number of agents grows, often via primal-dual analysis. Primal-dual techniques, despite producing celebrated results, have left gaps in our understanding of the \waterfill principle and equitable dynamic resource allocation. In particular, why is greedily equalizing resources seemingly the only viable strategy for many applications of dynamic allocation in worst-case settings with an arbitrary number of agents?

\subsection{Overview of Results and Techniques}
One can formulate worst-case analysis of dynamic resource allocation as
a two-player zero-sum game between an algorithm designer and an adversary.
The algorithm designer moves first by selecting a
(deterministic or randomized) policy $\calA$ for making allocation
decisions given the available information. The adversary
responds by constructing a sequence of requests $E$, where
each request designates a quantity of divisible resources
to be allocated and a set of eligible recipients. Running policy
$\calA$ on request sequence $E$ leads to a load
vector, $\calA(E)$, specifying the total quantity of
resources allocated to each agent. The game's payoff 
is the $\alpha$-regret: for a given maximization objective 
$f : \R^{n}_{\ge 0} \to \R_{\ge 0}$ and \emph{comparison
factor} $\alpha \ge 0$, this is given by the formula
\[
  \text{$\alpha$-Regret}
  = \alpha \cdot \max_{\ell \in \Delta(E)} f(\ell)
  - \mathbb{E}[f(\calA(E))],
\]
where $\Delta(E)$ denotes the set of 
feasible load vectors for request sequence $E$.
One can also formulate $\alpha$-regret for a
minimization objective $g : \R^{n}_{\ge 0} \to \R_{\ge 0}$,
using the formula
\begin{align*}
  \text{$\alpha$-Regret}
  = \mathbb{E}[g(\calA(E))] - \alpha \cdot \min_{\ell \in \Delta(E)} g(\ell)
\end{align*}
The adversary seeks to maximize $\alpha$-regret,
while the algorithm designer seeks to minimize it. 
Formulating worst-case analysis in terms of $\alpha$-regret is at
least as general as formulating it in terms of competitive analysis: an algorithm is $\alpha$-competitive if and only if the supremum
of its $\alpha$-regret is finite, and it is strictly $\alpha$-competitive
if and only if this supremum is non-positive.

Our main result pertains to objective functions that obey the following minimal requirement to be considered
equity-promoting: when one agent (Alice) receives 
weakly less resources than another (Bob), the objective value
cannot be improved by taking resources from Alice and giving
them to Bob. This property is called Schur-concavity (for
maximization objectives) or Schur-convexity (for minimization
objectives); we use the term \emph{Schur-monotonicity}
as a catch-all for both properties. 

We show that in the
$\alpha$-regret minimization game defined by \emph{any} 
Schur-monotone objective and \emph{any} non-negative comparison
factor $\alpha$, the \waterfill policy
is the optimal play for an algorithm designer against adaptive adversaries. 
When the adversary is restricted to be oblivious, 
\waterfill remains optimal
for the (narrower) class of symmetric concave maximization objectives
and symmetric convex minimization objectives. 
These results on the minimax optimality of \waterfill
persist even if the adversary is constrained to
generate request sequences with a specified number of
offline nodes,  number of requests, and 
total amount of resources to be allocated.
In other words, even though the \waterfill algorithm has
no foreknowledge of these parameters, nor of the objective function, it
is minimax optimal even when compared against algorithms that
have such foreknowledge.

In summary, our results show that the worst-case analysis
of dynamic equitable allocation for divisible resources
produces surprisingly robust guidance for algorithm designers:
\waterfill is optimal against adaptive adversaries for essentially
all equity-promoting objectives, and it is optimal against oblivious
adversaries for a broad sub-class of such objectives.
In particular,
Schur-monotone objective functions include
the most common objective functions for
evaluating fractional allocations when agents
are treated symmetrically. In fractional matching
when all agents have equal capacity, $c > 0$, the objective
function is $f(\mbf{x}) = \sum_{i=1}^n \min(c,\mbf{x}(i))$,
which is Schur-concave. The Nash 
social welfare objective, $\left( \prod_{i=1}^n \mbf{x}(i) \right)^{1/n}$,
is likewise Schur-concave. Common minimization
objectives such as $p$-norms and the Gini index are Schur-convex. 
In~\Cref{ssec:majorization} we give several examples of widely-studied objectives that fall in this class.

To derive worst-case optimality results pertaining to
such a broad range of objectives, we abstract
away from objective functions altogether and
directly compare load vectors using
\textit{majorization}, a preorder relation that compares the equitability of different allocations.
Schur-convex (resp., Schur-concave) functions are
precisely the order-preserving (resp., order-reversing)
functions from the majorization order to the
real numbers. Underpinning our results stated above
is a more fundamental result, \Cref{thm:wf_min}, that
expresses the optimality of the \waterfill policy
directly in terms of the majorization relation.
Since majorization is a preorder and not a total order, a given set of vectors need not have a minimal or a maximal element. However, the worst-case analysis
of \waterfill is facilitated by the existence
of such majorization-extremal elements in two
important cases. 
First, in \Cref{fact:opt}, we show that the set of feasible load vectors for a given request sequence $E$ always
has a unique majorization-minimal element, which we denote
by $\opt(E)$. Second, in \Cref{thm:param} we prove
that for any specified hindsight-optimal vector
$\mbf{\ell}$, as $E$ ranges over all the request sequences
with $\opt(E)=\mbf{\ell}$, the set of load vectors
$\wf(E)$ produced by \waterfill has a majorization-maximal
element, and we give an explicit 
formula for this element.
The minimax optimality of \waterfill
among deterministic algorithms is established via  
\Cref{thm:wf_min},
which asserts that for any request sequence $E$, if an algorithm designer deviates from \waterfill to some other deterministic allocation policy $\calA$, then the adversary can respond by modifying $E$ to a
request sequence $E'$ that produces a less equitable load vector --- $\E\calA(E')$ majorizes $\wf(E)$ ---
but a more equitable hindsight-optimal solution, in that
$\opt(E')$ is majorized by $\opt(E)$.
(If $\calA$ is randomized and the adversary
is allowed to be adaptive, we show in \Cref{app:sec:adapt_adv}
that the same conclusion
holds with probability 1 instead of with expectations.)
Regardless of the choice of Schur-monotone
objective function, the adversary's response to the algorithm designer's deviation makes
the objective value worse at $\mathbb{E}[\calA(E')]$
and better at $\opt(E')$, implying that worst-case
$\alpha$-regret increases when deviating from $\wf$
to any $\calA$. 
In the case of randomized $\calA$ and oblivious adversaries, 
an application of Jensen's inequality leads to the same conclusion
provided the objective function is symmetric
and concave (or convex).

To prove \Cref{thm:wf_min} we  take the perspective of an adversarial environment that designs the availability/compatibility of resources to hinder policies trying to produce majorization minimal allocations.
The first part of the proof (\Cref{sec:nest}) identifies a class of worst-case request sequences for \waterfill, which are represented by \textit{nested edge sets} on a bipartite graph.
Our proof uses a series of combinatorial operations (\Cref{alg:nest}) to transform a given request sequence into one encoded by a nested sequence on which \waterfill provides a majorizing allocation, and the hindsight optimal solution is minorizing.
The second part of our proof (\Cref{sec:wf_vs_policy}) considers a principal deviating from \waterfill to play an arbitrary policy.
We show that under such a deviation, the alternative policy produces a less equitable expected allocation than \waterfill on request sequences encoded by a nested edge set.
Combining these results, we find that \waterfill yields a majorization minimal expected allocation when requests are generated by a worst-case adversary.

\subsection{Related Work}\label{subsec:related_work}

\noindent\textbf{Equitable Allocations and Majorization}
Majorization is a well-established tool for equitable allocation in economics \cite{dutta1989concept} and operations research \cite{veinott1971least, megiddo1974optimal}. More recently, it has been applied to offline integral allocation \cite{harvey2003semi}, approximate online load balancing \cite{bhargava2001using, goel2005approximate, kumar2006fairness}, and a variety of one-shot strategic decision-making settings~\cite{banerjee2024fair, banerjee2025majorized, bai2025fair}. 
Our work follows in this line, extending this theory to online fractional allocation.

\noindent\textbf{Online Matching and Generalizations} 
The seminal work of~\citet{karp1990optimal} on online bipartite matching has led to a vast follow-up literature; see~\cite{mehta2013online} for a survey. Much of this focuses on integral allocations, and uses online primal-dual methods~\cite{devanur2013randomized,buchbinder2009design}.
The most relevant work for us is \cite{feige2020tighter}, which analyzes \waterfill for fractional matching and identifies the complete upper-triangular graph as the worst-case input among instances containing a perfect matching. 
Some of our arguments in~\Cref{sec:nest} closely mirror those in \cite{feige2020tighter}, and in a sense, by using majorization, our work generalizes his result to a wide variety of settings beyond matching, and also suggests natural extensions of the minimax conjecture of~\cite{karp1990optimal} (though these are still beyond our techniques).

\noindent\textbf{Fractional Allocation and Concave Utilities} A more closely related line is that initiated by~\citet{devanur2012online} on online allocation with concave utilities. The original work provided asymptotically-optimal competitive ratios via solutions to certain differential equations. 
More recent works~\cite{hathcock2024online,patton2026concave} provide asymptotic $1-1/e$ bounds for a wider range of objectives.
All these results use dual fitting approaches, which can often obscure the underlying intuition.
Indeed,~\citet{devanur2012online} explicitly comment that the exact bounds are ``a bit of a mystery''. We hope our work sheds some light on how these specific bounds emerge.

\noindent\textbf{\waterfill and Scheduling} Our work also relates to the literature on \emph{universal guarantees} in  scheduling, where a single policy simultaneously does well across multiple objectives. This was pioneered by~\cite{kleinberg1999fairness}, who characterized prefix-based fairness measures via majorization. More recent work achieves simultaneous competitiveness for all $\ell_p$ norms~\cite{chakrabarty2019approximation, chakrabarty2019simpler,kesselheim2023online}, and stronger guarantees using predictions~\cite{cohen2023general}. The last of these noted their approach had parallels to work on online Nash welfare maximization~\cite{banerjee2022online}; our work suggests a rationale behind this observation. These works all consider more complex constraint spaces than ours, and hence have weaker guarantees.

\noindent\textbf{Exact Minimax-Optimal Online Algorithms} 
Finally, we briefly comment on the rarity of \emph{exact minimax solutions} in online algorithms -- results where an algorithm's performance perfectly matches the theoretical value of the game on input sequences of bounded length.  
Notable examples include online prediction with 2 or 3 `experts'~\cite{cover1965binary, gravin2016towards}, unconstrained linear optimization~\cite{mcmahan2013minimax}, ski-rental~\cite{karlin1994competitive}, and paging \cite{sleator1985amortized,achlioptas2000competitive}. 
We situate \waterfill for online fractional allocation within this select family, noting its unique complexity and universality.

\section{Preliminaries}
\label{sec:def}

We denote the set $[n]\define\{1,\dots,n\}$ and $[0]\define \emptyset$, and use the shorthand $\pp{x} \define \max(x, 0)$, and $\1{p}$ for the indicator function that outputs $1$ when predicate $p$ is true and $0$ otherwise.
For request sequence $E=((N_t, q_t))_{t\in[m]}$ we denote the neighbors of offline node $i$ as $\Gamma_i(E) \define \{t\in[m] \mid i\in N_t\}$.

Let $\mbf{x}, \mbf{y} \in\R_{\ge 0}^n$ be arbitrary vectors.
We use $\mbf{x}^{\downarrow}$ and $\mbf{x}^{\uparrow}$ to denote vectors whose $i^{th}$ component is the $i^{th}$ largest and $i^{th}$ smallest element of $\mbf{x}$, respectively.
For any subset $N\subseteq [n]$, we write $\mbf{x} =_N \mbf{y}$ (similarly $\mbf{x} \le_N \mbf{y}$) to mean $\mbf{x}(i) = \mbf{y}(i)\,\forall\,i\in N$ (and $\mbf{x}(i) \le \mbf{y}(i) \,\forall\,i\in N$, respectively), and also denote the sum of elements in $N$ by $\mbf{x}(N) \define \sum_{i\in N} \mbf{x}(i)$.

\subsection{Setting and Problem Statement}\label{subsec:setting}
\label{ssec:setting}

\subsubsection{The Allocation Game}
We consider a zero-sum game between an algorithm designer and an adversary on a system with a set of $n$ offline nodes (agents) and $m$ online nodes (resources) that arrive sequentially.
Each online node $t \in [m]$ is characterized by a neighborhood $N_t \subseteq [n]$ of compatible offline agents and a divisible quantity $q_t > 0$, which the adversary designs.
A \emph{request sequence} is denoted by $E = ((N_t, q_t))_{t \in [m]}$, and we use $\calE_{n,m,q}$ to represent the set of all sequences with total capacity $q = \sum_t q_t$.
To consider all values for some characteristic, we replace it with ``$:$'' -- for instance, $\calE_{n,:,q} = \bigcup_{m\in \N} \calE_{n,m,q}$ represents sequences with $n$ offline nodes and total capacity $q$, but split between an arbitrary number of online nodes.

Upon the arrival of online node $t$, the algorithm designer must immediately select a fractional allocation $\mbf{x}_t \in \R_{\ge 0}^n$.
She selects her allocation via a (possibly randomized) policy $\calA$, i.e.~a function that maps the current node and the history of arrivals and decisions to a feasible allocation. 
An allocation is feasible if it satisfies:
\begin{itemize}
    \item \textbf{Compatibility:} $\mbf{x}_t(i) = 0$ for all $i \notin N_t$.
    \item \textbf{Full Distribution:} $\sum_{i \in [n]} \mbf{x}_t(i) = q_t$.
\end{itemize}
Full distribution is necessary for minimization objectives (e.g. in scheduling/load balancing), and without loss for non-decreasing maximization objectives (i.e., settings with `free disposal').

For any request sequence $E = ((N_t, q_t))_{t \in [m]}$, we use $\Delta(N_t, q_t)$ to denote the set of feasible allocations for online node $t$.
The Minkowski sum of $\Delta(N_t, q_t)$ is $\Delta(E)$, and this set contains every feasible load vector on $E$.
We also define $\calA(E) = \sum_{t=1}^{m} \mbf{x}_t$ as the cumulative allocation produced by policy $\calA$ on $E$, and $\bell_t = \sum_{s=1}^{t} \mbf{x}_s$ as the cumulative intermediate allocation (or load-vector) after $t$ arrivals.

Finally, we identify a special subset of request sequences that are critical for our characterization.

\begin{definition}[Nested sequences]
\label{def:nested_edge_set}
    A request sequence $E=((N_t, q_t))_{t\in[m]}$ is said to be nested if $N_1 \supseteq \dots \supseteq N_m$, i.e., online neighborhoods form a chain, with each subsequent neighborhood being nested within the previous neighborhoods.
    $\Enest_{n,m,q}\subseteq \calE_{n,m,q}$ is the set of nested sequences with $n$ offline nodes, $m$ online nodes, and total quantity $q$.
\end{definition}

\subsubsection{Objectives and Competitive Performance}
The goal of the algorithm designer for policy $\calA$ is to optimize some objective function $\E f(\calA(E))$.
Moreover, the performance of the policy $\calA$ is measured against an optimal-in-hindsight solution which knows the entire sequence $E$ in advance. 
We formalize this via the notion of $\alpha$-regret, which generalizes both regret and the commonly considered competitive ratio objectives.
In more detail, given a comparison factor $\alpha\in R_{>0}$, we consider two types of objectives:
\begin{enumerate}
    \item \textbf{Fair Welfare Maximization:} For any equity/welfare measure $f:\R_{\ge 0}^n \to \R_{\ge 0}$, the principal maximizes $\E[f(\calA(E))]$ against a discounted optimal hindsight solution:
    \[
        c_{\alpha, f}^{\max}(\calA, E) \define \alpha \cdot \max_{\bell\in\Delta(E)}f(\bell) - \E [f(\calA(E))]
    \]
    \item \textbf{Load Balancing:} For any load/inequity measure $g:\R_{\ge 0}^n \to \R_{\ge 0}$, the principal minimizes $\E[g(\calA(E))]$ against a marked-up optimal hindsight solution:
    \[
        c_{\alpha, g}^{\min}(\calA, E) \define \E [g(\calA(E))] - \alpha \cdot \min_{\bell\in\Delta(E)} g(\bell)
    \]
\end{enumerate}

The game between the principal and the adversary is parameterized by offline node count $n\in\N$, online node count $m\in\N$, total quantity $q\in\R_{> 0}$, objective $f$ (or $g$), and comparison factor $\alpha\in R_{>0}$.
The algorithm designer moves before the adversary and selects a policy $\calA$.
The adversary responds with a request sequence $E\in\calE_{n,m,q}$, which they design with access to the policy $\calA$ played by the algorithm designer (but not its random coin flips).
In other words, we model the adversary as oblivious. See 
\Cref{app:sec:adapt_adv} for an extension of our model and results
to adaptive adversaries.
The goal of the algorithm designer is to minimize her worst-case $\alpha$-regret:
\begin{equation*}
    \reg_{(n,m,q),\alpha, f}^{\max}(\calA) \define \sup_{E \in \calE_{n,m,q}} c_{\alpha,f}^{\max}(\calA, E)
    \quad \mbox{or} \quad
    \reg_{(n,m,q),\alpha, g}^{\min}(\calA) \define \sup_{E \in \calE_{n,m,q}} c_{\alpha,g}^{\min}(\calA, E).
\end{equation*}
Comparison factor $\alpha=1$ corresponds to standard regret for policy $\calA$.
On the other hand, the competitive ratio of $\calA$ corresponds to the largest (resp. smallest) $\alpha$ such that $\reg_{(n,m,q),\alpha, f}^{\max}(\calA) \leq 0$ (resp. $\reg_{(n,m,q),\alpha, g}^{\min}(\calA)\leq 0$). 
Our results encompass all of these settings (parameters $(m,n,q,\alpha)$, both objective types, and a large class of functions $f$) in a universal manner via a single policy. 

\subsection{Majorization}
\label{ssec:majorization}

Unlike previous work that studies different objective functions applied to fractional allocation vectors, we directly compare and analyze the vectors using majorization.
This preorder relation compares the equity of two given vectors, with vectors lower in the order being more equitable.
The economic justification behind majorization is the so-called Pigou-Dalton (or Robin Hood) transfer principle, which states that transferring resources from richer to poorer agents yields a more equitable distribution.
Below, we define the parts of the theory necessary for our work; for a more detailed exposition, refer to~\cite{arnold1987majorization, marshall2011majorization}.

\begin{definition}[Majorization]\label{def:maj}
    Let $\mbf{x},\mbf{y}\in \R^n_{\ge 0}$ with $\mbf{x}([n]) = \mbf{y}([n])$.
    Then we say $\mbf{x}$ majorizes $\mbf{y}$ (denoted $\mbf{x}\succeq\mbf{y}$) iff the cumulative sums in decreasing order for $\mbf{x}$ dominate those for $\mbf{y}$: 
    $$
        \mbf{x}^{\downarrow}([j]) \ge \mbf{y}^{\downarrow}([j]) \quad\forall\, j\in [n]
    $$
    Equivalently, for cumulative sums in increasing order, we have $\mbf{x}^{\uparrow}([j]) \le \mbf{y}^{\uparrow}([j])$ for all $j\in [n]$.
\end{definition}

We sometimes find it convenient to say $\mbf{y}$ minorizes $\mbf{x}$ (and write $\mbf{y} \preceq \mbf{x}$).
Also, if vectors $\mbf{x}, \mbf{y} \in\R_{\ge 0}^n$ majorize each other (i.e. $\mbf{x}\succeq \mbf{y}$ and $\mbf{x} \preceq \mbf{y}$), then we say they are in the same equivalence class, and write $\mbf{x} \sim \mbf{y}$.
Vectors $\mbf{x},\mbf{y}$ are equivalent iff there exists a permutation matrix $P$ satisfying $\mbf{x} = P \mbf{y}$. Below are equivalent conditions for $\mbf{x}\succeq\mbf{y}$.

\begin{fact}[Equivalent Characterizations of Majorization]
\label{def:KaramataHLP}
    Let $\mbf{x},\mbf{y}\in \R^n_{\ge 0}$ with $\mbf{x}([n]) = \mbf{y}([n])$.
    Then $\mbf{x}\succeq\mbf{y}$ if and only if the following equivalent conditions hold:
    \begin{itemize}
        \item Hardy, Littlewood, and Polya \cite{hardy1988inequalities}: $\mbf{y} = A \mbf{x}$ for some double stochastic $A\in\mbf{R}_{\ge 0}^{n\times n}$.
        \item Karamata's Inequality\footnote{We state a special case of Karamata's theorem. When $\mbf{x}([n])=\mbf{y}([n])$, it suffices for the inequality to hold on thresholding functions, whereas consideration of all convex functions is required in general.}: $\sum_{i\in [n]} \pp{\mbf{x}(i) - \gamma} \ge \sum_{i\in[n]} \pp{\mbf{y}(i) - \gamma}$ for every threshold $\gamma\in\R_{\ge 0}$.
    \end{itemize}
\end{fact}

\subsubsection{Schur-Monotone Functions}

Majorization is closely related to Schur-monotone functions: Schur-concave/Schur-convex functions reverse/follow the majorization order, respectively.
\begin{definition}[Schur-monotone Functions]
    A function $f:\R^n\to\R$ is Schur-concave if for all $\mbf{x,y}\in\R^n$ such that $\mathbf{x} \preceq \mathbf{y}$, we have $f(\mathbf{x}) \ge f(\mathbf{y})$.
    In contrast, a Schur-convex function $g:\R^n\to\R$ has the property that $\mathbf{x} \preceq \mathbf{y}$ implies $g(\mathbf{x}) \le g(\mathbf{y})$.
\end{definition}

The Schur-concave and Schur-convex function classes contain most meaningful measures of equity and inequity, respectively.
These classes are expressive; for instance, every symmetric and quasi-convex (resp. quasi-concave) function is Schur-convex (resp. Schur-concave)~\cite[Proposition C.3]{marshall2011majorization}.
We provide a few notable examples; many more are available in \cite{arnold1987majorization, marshall2011majorization}.
In each of the following, we take $\mbf{x}\in\R_{\ge 0}^n$ and $\bar{\mbf{x}} = \frac{1}{n}\sum_{i\in[n]} \mbf{x}(i)$.

\medskip

\noindent
\begin{minipage}[t]{0.48\textwidth}
    \textbf{Schur-concave Functions}
    \begin{itemize}[leftmargin=*]
        \item \textbf{Matching:}  $\sum_{i\in[n]}\min(c, \mbf{x}(i))$ for $c > 0$
        \item \textbf{Nash Social Welfare:}  $\left(\prod_{i\in [n]} \mbf{x}(i)\right)^{1/n}$
        \item \textbf{Egalitarian Welfare:} $\min_{i\in[n]}{x[i]}$
        \item \textbf{Power Means:}  $\left(\sum_{i\in[n]}|\mbf{x}(i)|^p\right)^{1/p}$ for $p < 1$
    \end{itemize}
\end{minipage}
\hfill
\begin{minipage}[t]{0.48\textwidth}
    \textbf{Schur-convex Functions}
    \begin{itemize}[leftmargin=*]
        \item \textbf{Gini Index:}  $\frac{1}{2n^2 \bar{\mbf{x}}}\left( \sum_{i,j\in[n]} |\mbf{x}(i) - \mbf{x}(j)| \right)$
        \item \textbf{Variance:}  $\sum_{i\in[n]} \left(\mbf{x}(i) - \bar{\mbf{x}}\right)^2$
        \item \textbf{Makespan:} $\max_{i\in[n]}{x[i]}$
        \item \textbf{$\ell^p$-norms:}  $\left(\sum_{i\in[n]}|\mbf{x}(i)|^p\right)^{1/p}$ for $p \ge 1$
    \end{itemize}
\end{minipage}

\subsubsection{Majorization Minimal Offline Allocations}
\label{ssec:offlineopt}

For online algorithms, typical performance measures compare the objective value achieved by an online policy to the optimal offline objective (i.e., that of the best allocation in hindsight).
It turns out that each request sequence $E\in\calE_{n,m,q}$ admits a unique load vector that is majorization minimal.
This allocation thus serves as a universal optimal hindsight solution, as it simultaneously maximizes all Schur-concave functions and minimizes all Schur-convex functions.

\begin{fact}[Existence of a Majorization Minimal Hindsight Solution]\label{fact:opt}
        Given a request sequence $E = ((N_t, q_t))_{t\in[m]}$, there is a unique allocation (load vector) $\bell^*$ that is minimal in the majorization preorder: $\bell^* \preceq \bell$ for any feasible allocation $\bell\in\Delta(E)$.
        Consequently, $\bell^* \in \arg\max_{\bell\in\Delta(E)} f(\bell)$ for all Schur-concave $f$ and $\bell^*\in \arg\min_{\bell\in\Delta(E)} g(\bell)$ for all Schur-convex $g$.
\end{fact}

Based on this, we use the shorthand $\opt(E) \define \bell^*$ for the optimal hindsight solution.
For completeness, we provide a brief proof outline for~\Cref{fact:opt} in~\Cref{appsec:proofs}.
In \Cref{fig:opt_unified}, we provide an example of a request sequence and the corresponding optimal hindsight allocation.
We return to this sequence in later illustrations (\Cref{fig:wf,fig:nestification,fig:oper}).

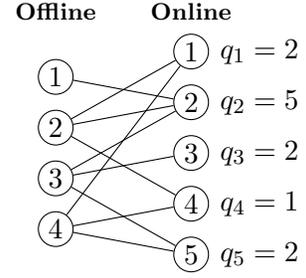
\begin{figure}[t]
    \centering
    \begin{minipage}[c]{0.62\textwidth}
        \centering
        \small
        \resizebox{\linewidth}{!}{
            \begin{tabular}{|c|l|cccc||cccc|cccc|}
                \hline
                \multicolumn{2}{|c|}{Quantities}
                & \multicolumn{4}{c||}{\makecell[c]{Sequence}}
                & \multicolumn{4}{c|}{\text{OPT} Alloc.}
                & \multicolumn{4}{c|}{\text{OPT} Loads} \\
                \hline
                \multirow[c]{5}{*}{\rotatebox{90}{Resources}}
                & $q_1=2$
                &     & $1$ &     & $1$
                &     & $1$ &     & $1$ 
                & $0$ & $1$ & $0$ & $1$
                \\
                & $q_2=5$
                & $1$ & $1$   & $1$ & 
                & $3$ & $1$   & $1$ & 
                & $3$ & $2$   & $1$ & $1$ 
                \\
                & $q_3=2$
                &     &     & $1$ & 
                &     &     & $2$ & 
                & $3$ & $2$ & $3$ & $1$
                \\
                & $q_4=1$
                &     & $1$ &     & $1$
                &     & $1$ &     & $0$
                & $3$ & $3$ & $3$ & $1$
                \\
                & $q_5=2$
                &     &     & $1$ & $1$
                &     &     & $0$ & $2$
                & $3$ & $3$ & $3$ & $3$
                \\
                \hline
            \end{tabular}
        }
    \end{minipage}
    \hfill
    \begin{minipage}[c]{0.35\textwidth}
        \centering
        \begin{tikzpicture}[
            every node/.style={circle, draw, minimum size=1.2mm, inner sep=1.5pt},
            node distance=12mm,
            scale=0.45
        ]
    
        \node[draw=none] at (0,7.2) {\textbf{\footnotesize Offline}};
        \node[draw=none] at (4,7.2) {\textbf{\footnotesize Online}};
        
        \node (u1) at (4,6) {1};
        \node (u2) at (4,4.5) {2};
        \node (u3) at (4,3) {3};
        \node (u4) at (4,1.5) {4};
        \node (u5) at (4,0) {5};
    
        \node[draw=none, right=0.5mm of u1] {$q_1=2$};
        \node[draw=none, right=0.5mm of u2] {$q_2=5$};
        \node[draw=none, right=0.5mm of u3] {$q_3=2$};
        \node[draw=none, right=0.5mm of u4] {$q_4=1$};
        \node[draw=none, right=0.5mm of u5] {$q_5=2$};
        
        \node (v1) at (0,5.25) {1};
        \node (v2) at (0,3.75) {2};
        \node (v3) at (0,2.25) {3};
        \node (v4) at (0,0.75) {4};
        
        \draw (u1) -- (v2); \draw (u1) -- (v4);
        \draw (u2) -- (v1); \draw (u2) -- (v2); \draw (u2) -- (v3);
        \draw (u3) -- (v3);
        \draw (u4) -- (v2); \draw (u4) -- (v4);
        \draw (u5) -- (v3); \draw (u5) -- (v4);
        
        \end{tikzpicture}
    \end{minipage}
    \caption{\small\em
        Example of a request sequence $E = ((N_t, q_t))_{t\in[5]}\in\calE_{4,5,12}$ with $n=4$ offline nodes, $m=5$ online nodes and quantity $q=12$.
        The table gives the quantities and adjacency matrix between offline (columns) and online (rows) nodes, and details the allocation $\opt(E)$, and intermediate loads for offline nodes over time. 
        The diagram on the right depicts the request sequence, with online nodes arriving from top to bottom.
    }
    \label{fig:opt_unified}
\end{figure}

\subsection{Water Filling}
\label{subsec:def_wf}
We now formally define the \waterfill algorithm. 
Upon arrival of an online node $t$, \waterfill fractionally allocates the arriving resource $q_t$ to maximize the minimal cumulative allocation of neighboring offline nodes $N_t$.

\begin{definition}[\waterfill]\label{def:wf}
        The (deterministic) \waterfill policy proceeds as follows:
        after producing allocations $\mbf{x}_1,\dots,\mbf{x}_{t-1}$ on the first $t-1$ arrivals to yield intermediate load vector $\bell_{t-1} = \sum_{s\ge t-1} \mbf{x_s}$, on the arrival of the $t^{th}$ online node $(N_t,q_t)$,   \waterfill selects the allocation $\mbf{x}_t$ which solves the following convex program:

    \begin{align*}
        &\max \min_{i\in N_t} \left\{\mbf{x}_t(i) + \bell_{t-1}(i)\right\} 
        \quad\quad\text{s.t.} \quad \mbf{x}_t\in\Delta(N_t, q_t)
    \end{align*}
\end{definition}

In~\Cref{fig:wf}, we show the \waterfill allocations on the example sequence from~\Cref{fig:opt_unified}.
We use $\wf(N_t, q_t, \bell_{t-1})$ to denote the allocation of \waterfill for the $t^{th}$ resource, and $\wf(E) \define \bell_m$ as the final allocation produced by \waterfill on request sequence $E\in\calE_{n,m,q}$.

\noindent Our results also rely on the following properties of \waterfill, which are easy to verify:
\begin{restatable}[Properties of \waterfill]{observation}{wfproperties}
\label{obs:wf_properties}
    For any non-empty neighbor set $N\subseteq [n]$, quantity $q>0$, and load-vector $\bell\in\R_{\ge 0}^n$, the allocation $\mbf{x}^*$ selected by \waterfill satisfies:
    \begin{itemize}[leftmargin=*]
        \item \textbf{Uniqueness:} $\mbf{x}^* + \bell$ is the unique majorization minimal vector in the feasible set $\{\mbf{x} + \bell\}_{\mbf{x}\in \Delta(N,q)}$.
        \item \textbf{Equalized Loads:} $\forall\,i, j$ in support $\widehat{N} \coloneqq \{i\in N \mid \mbf{x}^*(i) > 0\}$, we have $\mbf{x}^*(i) + \bell(i) = \mbf{x}^*(j) + \bell(j)$.
    \end{itemize}
\end{restatable}

\begin{figure}[t]
    \centering
    \begin{tabular}{|c|l|cccc||cccc|cccc|}
        \hline
        \multicolumn{2}{|c|}{Quantities}
        & \multicolumn{4}{c||}{\makecell[c]{Sequence}}
        & \multicolumn{4}{c|}{\wf\, Alloc.}
        & \multicolumn{4}{c|}{\makecell[c]{\wf\, Loads}} \\
        \hline
        \multirow[c]{5}{*}{\rotatebox{90}{Resources}}
        & $q_1=2$
        &     & $1$ &     & $1$
        &     & $1$ &     & $1$ 
        & $0$ & $1$ & $0$ & $1$
        \\
        & $q_2=5$
        & $1$ & $1$ & $1$ & 
        & $2$ & $1$ & $2$ &
        & $2$ & $2$ & $2$ & $1$ 
        \\
        & $q_3=2$
        &     &     & $1$ &
        &     &     & $2$ &
        & $2$ & $2$ & $4$ & $1$
        \\
        & $q_4=1$
        &     & $1$ &     & $1$
        &     & $0$ &     & $1$
        & $2$ & $2$ & $4$ & $2$
        \\
        & $q_5=2$
        &     &     & $1$ & $1$
        &     &     & $0$ & $2$
        & $2$ & $2$ & $4$ & $4$
        \\
        \hline
    \end{tabular}
    \caption{\small\em
        \waterfill allocations on the~\Cref{fig:opt_unified} sequence.
        The columns from left to right show the quantity and neighbors of each online node, allocations made by $\waterfill$, and intermediate load vectors.
    }\label{fig:wf}
\end{figure}

\section{Main Results}
\label{sec:result}

We now state our main results, establishing that the adversary's best response is a nested sequence are a centerpiece of our analysis, as they induce majorization maximal \waterfill allocations.

\subsection{Minimax Optimality of \waterfill}

We first prove that \waterfill minimizes $\alpha$-regret compared to any other policy in a minimax sense for a wide range of settings.
This fact is a corollary of our following main theorem, which states that \waterfill is ``majorization minimal'' in worst-case environments.

\begin{restatable}[\waterfill is Majorization Minimal]{theorem}{wfmin}
\label{thm:wf_min}
    Fix parameters $n,m\in\N$ and $q\in\R_{> 0}$. Given any (possibly randomized) allocation policy $\calA$ and any request sequence $E\in\calE_{n,m,q}$, there is a nested sequence $E'\in\Enest_{n,m,q}$ for which $\wf(E) \preceq \E[\calA(E')]$ and $\opt(E) \succeq \opt(E')$.
\end{restatable}
The expectation above is with respect to the randomness of $\calA$.
This result follows immediately from composing the following two lemmas.
The first lemma provides a constructive procedure (\Cref{alg:nest}) that takes as input a request sequence $E$, and identifies a nested sequence with the same input parameters on which \waterfill yields a majorizing allocation. This establishes that the adversary's best response to the \waterfill policy is always a nested request sequence.%
\setcounter{lemma}{0}%
\begin{restatable}[Majorization Maximality of Nested Sequences]{lemma}{worstcasewf}
\label{lem:nest_wf}
Given any sequence $E\in\calE_{n,m,q}$, \Cref{alg:nest} returns a nested sequence $\wtE\in\Enest_{n,m,q}$ satisfying $\wf(E) \preceq \wf(\wtE)$ and $\opt(E) \succeq \opt(\wtE)$.
\end{restatable}
The second lemma shows that other allocation policies also perform poorly on nested request sequences.
We provide a constructive procedure (\Cref{alg:wf_vs_policy}) that takes as input an arbitrary policy $\calA$ and a nested sequence, and identifies a second nested sequence with equivalent hindsight optimal allocation; further, the \waterfill allocation on the original sequence minorizes the $\calA$ allocation on the new sequence.
Thus, an algorithm designer that deviates from \waterfill to an alternative policy performs worse against an adversary restricted to playing nested sequences.

\setcounter{lemma}{1}
\begin{restatable}[Policy Deviation]{lemma}{wfvspolicy}
\label{lem:wf_vs_policy}
    For any policy $\calA$ and nested sequence $\wtE\in\Enest_{n,m,q}$, we can transform it (via \Cref{alg:wf_vs_policy}) into a nested sequence $E'\in\Enest_{n,m,q}$ with $\wf(\wtE)\preceq\E[\calA(E')]$ and $\opt(\wtE)\sim\opt(E')$.
\end{restatable}

For deterministic allocation policies, \Cref{thm:wf_min} implies that \waterfill, \emph{simultaneously for every $\alpha$}, minimizes $\alpha$-regret (and hence maximizes the competitive ratio) for the goal of  \emph{minimizing/maximizing any Schur-convex/Schur-concave function}.
Formally, \waterfill is a minimax-optimal in the class of deterministic policies in any instantiation of our online allocation game.%
\begin{restatable}[Deterministic Minimax Optimality of \wf]{corollary}{determinimax}\label{cor:reg_det}
    Fix $n,m,q,\alpha$, and a Schur-concave objective $f:\R^{n}_{\ge 0}\to\R_{\ge 0}$ (or Schur-convex objective $g:\R^{n}_{\ge 0}\to\R_{\ge 0}$).
    The $\alpha$-regret of $\wf$ in comparison to that of any deterministic allocation policy $\calA$ satisfies:
    \begin{align*}
        \reg_{(n,m,q),\alpha,f}^{\max}(\wf) \le \reg_{(n,m,q),\alpha,f}^{\max}(\calA)
        \quad\quad\text{or}\quad\quad
        \reg_{(n,m,q),\alpha,g}^{\min}(\wf) \le \reg_{(n,m,q),\alpha,g}^{\min}(\calA)
    \end{align*}
\end{restatable}
Importantly, this means that the principal cannot do better than \waterfill even with side information about the number of offline nodes\footnote{
    The principal need not know the number of offline nodes $n$ up front.
    Instead, she may infer the existence of an offline node when it appears as a neighbor of an online node.
} $n$, the number of online nodes $m$, the total quantity $q$, the particular objective function (within the class of Schur-convex/concave functions), or the performance measure ($\alpha$-regret/competitive ratio). 
Note also that \Cref{cor:reg_det} holds for randomized allocation policies if we consider ex-ante performance measures, i.e., define the cost incurred by the algorithm designer as $\alpha\cdot  f(\opt(E)) - f(\E [\calA(E)])$.
In \Cref{app:sec:adapt_adv}, we explain how to extend our model to adaptive adversaries; against such adversaries, \waterfill is \emph{ex-post} minimax optimal for all Schur-monotone objectives in comparison to all (possibly randomized) policies $\calA$.

For comparing against the ex-post optimal solution $\E[f( \calA(E))]$, the situation is more subtle.
Our next result provides a separation between deterministic and randomized policies by providing an objective for which \waterfill is not the minimax optimal randomized policy.

\begin{restatable}[Deterministic and Randomized Separation]{theorem}{needconcave}\label{thm:no_opt_rand}
    Fix number of offline nodes $n=2$, the number of online nodes $m=2$, and total quantity $q=2$, and consider the Schur-concave objective
    \begin{align*}
        f(\mbf{x}) \define \1{\mbf{x}(1) > 1/2 \text{ and }\mbf{x}(2) > 1/2}
    \end{align*}
    There is a randomized policy $\calA$ with lower $\alpha$-regret than \waterfill; in particular for all $\alpha > 0$:
    \begin{align*}
        \reg_{(2,2,2),\alpha,f}^{\max}(\calA) \le \pp{\alpha - \frac{1}{2}} < 
        \alpha \leq \reg_{(2,2,2),\alpha, f}^{\max}(\wf)
    \end{align*}
\end{restatable}

\Cref{thm:no_opt_rand} shows that for the expected ex-post objectives, \waterfill is not necessarily $\alpha$-regret minimizing among randomized policies for general Schur-concave objectives.
Thus, for \waterfill to compete against randomized policies, we need some additional assumption. The next result (again derived from~\Cref{thm:wf_min}) provides one such sufficient condition.

\begin{restatable}[Randomized Minimax Optimality of \wf]{corollary}{randminimax}\label{cor:reg_rand}
    For any $n,m,q,\alpha$, and any symmetric and \underline{concave} objective $f:\R_{\ge 0}^n\to\R_{\ge 0}$ (or any symmetric and \underline{convex} objective $g:\R_{\ge 0}^n\to\R_{\ge 0}$), the $\alpha$-regret of $\wf$ in comparison to that of any randomized allocation policy $\calA$ satisfies:
    \begin{align*}
        \reg_{(n,m,q),\alpha,f}^{\max}(\wf) \le \reg_{(n,m,q),\alpha,f}^{\max}(\calA)
        \quad\quad\text{or}\quad\quad
        \reg_{(n,m,q),\alpha,g}^{\min}(\wf) \le \reg_{(n,m,q),\alpha,g}^{\min}(\calA)
    \end{align*}
\end{restatable}
Before proceeding, we note that establishing~\Cref{lem:nest_wf,lem:wf_vs_policy} form the technical core of our work, and most of the main body of this paper is devoted to establishing these (\Cref{sec:nest,sec:wf_vs_policy} respectively). The remaining results follow from these lemmas in a straightforward way, and so their proofs are deferred to~\Cref{appsec:result}.

\subsection{Characterizing Worst-Case Sequences}
\label{ssec:uppertriangle}

\Cref{thm:wf_min} shows that \waterfill dominates any other policy in terms of $\alpha$-regret for any $\alpha$ and any (reasonable) equity-promoting objective.
In this section, we consider the adversary's side of the game and characterize their best response --- the worst-case input for \waterfill --- in~\Cref{thm:param}.
Informally, our result shows that given any number of offline nodes $n$ and total quantity $q$, it is sufficient for an adversary to only consider \emph{$n\times n$ complete upper-triangular} request sequences (i.e., with $m=n$ online nodes and $N_t=\{t,t+1,\ldots,n\}$) and \emph{non-decreasing quantities} (i.e., $q_1\leq q_2\leq\ldots\leq q_n$ and $\sum_{t=1}^nq_t=q$).

\begin{restatable}[Characterizing Adversarial Sequences]{theorem}{param}
\label{thm:param}
    Given a request sequence $\wtE\in\Enest_{n,m,q}$, we can transform it (\Cref{alg:param_nest}) into a nested sequence $E'=((N_t',q_t'))_{t\in[n]}$ satisfying all the following:
    \begin{enumerate}
        \item\label{item:construct} $E'\in\Enest_{n,n,q}$ (i.e., $n$ online nodes), with $N_t'=\{t,t+1,\ldots,n\}$ \hfill(Complete Upper Triangular)
        \item\label{item:order}
        $q_1'\leq q_2'\leq\ldots\leq q_n'$ and $\sum_{t=1}^nq_t'=q$ \hfill(Nondecreasing Quantities)
        \item\label{item:maj} $\opt(\wtE) \sim \opt(E')$, and $\wf(\wtE)\preceq \wf(E')$ \hfill(Dominates in Majorization Preorder)
    \end{enumerate}
\end{restatable}
\Cref{sec:uppertriangle} details how we obtain this characterization. 
Notably, given any $n,\alpha$, we can now formulate computing the minimax optimal $\alpha$-regret (and consequently, competitive ratio) as an optimization problem over $\R^n$. For brevity, we consider maximizing a Schur-concave function, but an analogous result holds for minimizing any Schur-convex function. 

\begin{restatable}[Minimax Regret Characterization]{corollary}{regret}
\label{cor:reg_bound}
Given any $n\in\N$, let $H\in \R_{\ge 0}^{n\times n}$ be the harmonic series matrix $H(i,j) = \frac{\1{i\ge j}}{n-j+1}$.
Then for any $\alpha >0$ and Schur-concave $f:\R_{\ge 0}^n\to\R_{\ge 0}$, the minimax optimal $\alpha$-regret satisfies:
    \begin{align*}
        \reg_{(n,:,:),\alpha, f}^{\max}(\wf) = \sup_{\bell\in\R^n: 0\leq \bell(1) \leq \bell(2) \leq \ldots \leq \bell(n)} \Big(\alpha \cdot f(\bell) - f(H \bell)\Big)
    \end{align*}
\end{restatable}
\noindent \Cref{table:cr} lists exact competitive ratios for various functions, obtained via this characterization. 

\begin{figure}[!ht]
    \centering
    \small
    \begin{tabular}{|c|c|c|c|}
        \hline
        Name  & Function & Max/Min & Competitive Ratio  \\
        \hline
        \hline
        \textsc{Nash Social Welfare} & $(\prod_i \mbf{x}(i))^{1/n}$ & Max & $(n!)^{-1/n}$ \\
        \hline
        \textsc{Minimax} & $\min_{i\in [n]}\mbf{x}(i)$ & Max & $1/n$ \\
        \hline
        \textsc{Maximin} & $\max_{i\in[n]}\mbf{x}(i)$ & Min & $\sum_{i\in[n]} 1/i$ \\
        \hline
        \textsc{Fractional Matching} & $\sum_{i\in[n]}\min(c, \mbf{x}(i))$, $c > 0$ & Max & $\min_{k\in[n]} M_k$ \\
        \hline
        \textsc{Separable Concave} & $\sum_{i\in[n]} \mathscr{f}(\mbf{x}(i))$, $\mathscr{f}$ concave & Max & $\ge \min_{k\in[n]} M_k$ \\
        \hline
    \end{tabular}
    \caption{\small\em
        The table depicts competitive ratios for various Schur-monotone objectives.
        The columns, from left to right, are the name and definition of the function, the form of objective (maximization/minimization), and the minimax competitive ratio for given $n$.
        Notably, the fractional matching objective with capacity $c=1$, $\fm_n(\mbf{x}) = \sum_{i\in[m]} \min(1, \mbf{x}(i))$, defines the $M_k$ sequence: $M_k = \frac{1}{k} \fm_k(H\vec{1})$. For details, see~\Cref{appsec:CR}.
    }
    \label{table:cr}
\end{figure}

\section{Water-Filling, Nested Sequences and Majorization}
\label{sec:nest}

In this section, we present an explicit procedure (\Cref{alg:nest}) that transforms an arbitrary request sequence $E$ into a nested one $\wtE$, such that the outcome of \waterfill moves higher in the majorization preorder from $E$ to $\wtE$, while that of the hindsight-optimal allocation moves lower. 

\worstcasewf*

This process and its analysis constitute
the most involved, but also the most illuminating, ingredient in the proof 
of \Cref{thm:wf_min}. 
It reveals that worst-case sequences for \waterfill have a nested structure, which greatly facilitates comparing 
\waterfill with other policies, as well as calculating competitive ratios.

\subsection{Proof Overview and Definitions}

A natural way to transform a sequence $E$ into a nested sequence $\wtE$ is to replace each neighbor set~$N_t$ with the union $N_t \cup N_{t+1} \cup \cdots \cup N_m$: we call this the nested request sequence induced by $E$: 
\begin{definition}[Induced Nested Sequence]\label{def:induce_nested}
    Fix an sequence $E = ((N_t, q_t))_{t\in [m]}\in\Enest_{n,m,q}$ and define the latest arriving online neighbor of offline node $i$: $\mu_i \define\max(\Gamma_i(E) \cup\{0\})$.
    The nested sequence induced by $E$ is $\wh{E} = ((\wh{N}_t, \wh{q}_t))_{t\in [m]} \in \Enest_{n,m,q}$ with $\wh{N}_t \gets \{i\in [n]\mid t \le \mu_i\}$ and $\wh{q}_t\gets q_t$. 
\end{definition}

\begin{figure}[ht]
    \centering
    \small
    \begin{tabular}{|c|l|cccc||cccc|}
        \hline
        \multicolumn{2}{|c|}{Quantities}
        & \multicolumn{4}{c||}{\makecell[c]{Sequence $E$}}
        & \multicolumn{4}{c|}{\makecell[c]{Nested \\ Sequence $\wh{E}$}} \\
        \hline
        \multirow[c]{5}{*}{\rotatebox{90}{Resources}}
        & $q_1=2$
        &     & $1$ &     & $1$
        & $1$ & $1$ & $1$ & $1$ 
        \\
        & $q_2=5$
        & \textcolor{red}{$1$} & $1$ & $1$ &  
        & $1$ & $1$ & $1$ & $1$
        \\
        & $q_3=2$
        &  &  & $1$ &
        &  & $1$ & $1$ & $1$
        \\
        & $q_4=1$
        &     & \textcolor{red}{$1$} &     & $1$
        &     & $1$ & $1$ & $1$
        \\
        & $q_5=2$
        &     &     & \textcolor{red}{$1$} & \textcolor{red}{$1$}
        &     &     & $1$ & $1$
        \\
        \hline
    \end{tabular}
    \caption{\small\em
        The table depicts the sequence $E\in\calE_{4,5,12}$ from~\Cref{fig:opt_unified}, and the corresponding induced nested sequence $\wh{E}\in\Enest_{4,5,12}$.
        Note that both sequences have the same number of offline nodes (columns) and online nodes (rows).
        For clarity, we highlight $\mu_i$, the latest arriving online neighbor of each offline node in red. 
    }
    \label{fig:nestification}
\end{figure}
One could hope that the relation $\wf(E) \preceq \wf(\wh{E})$ is satisfied by every sequence $E$; however, this is not the case. 
As a simple counterexample, note that whenever 
$N_m = [n]$, the sequence $\wh{E}$ is complete bipartite and 
$\wf(\wh{E})$ is the perfectly balanced vector $\frac{q}{n} \vec{1}$, 
which is majorization-minimal.

To get around this obstacle, we need to first characterize sufficient conditions on a sequence $E$ such that the 
relation $\wf(E) \preceq \wf(\wh{E})$ is true. For this, we need two additional definitions. First, we separate the edges of a sequence into those that are \emph{active} under \waterfill (i.e., in the support of the \waterfill allocation), and the rest.
\begin{definition}[Active and Inactive Edges]
    Fix a sequence $E = ((N_t, q_t))_{t\in[m]}\in\calE_{n,m,q}$ and an allocation $(\mbf{x}_t)_{t\in[m]}$ which is feasible on that sequence.
    An edge $(i,t)$ is said to be active when the allocation uses the edge $\mbf{x}_t(i) > 0$.
    Edges that are not active are called inactive.
\end{definition}
Secondly, using the equalizing property of \waterfill on active edges (see~\Cref{obs:wf_properties}), we define the notion of the \emph{height} of an online node under \waterfill.
\begin{definition}[Height of Online Node]\label{def:height}
    For sequence $E=((N_t, q_t))_{t\in[m]}$ with \waterfill allocations $(\mbf{x}_t)$, the \emph{height} of online node $t$ is $h_t \define \sum_{s=1}^{t}\mbf{x}_s(i)$ for any $i$ in support $\widehat{N}_t \coloneqq \{i \mid \mbf{x}^*_t(i) > 0 \}$.
\end{definition}
\noindent To see how these definitions relate to nested sequences, we can make the following observation: 
\begin{observation}[\waterfill on Nested Sequences]\label{obs:wf_on_nest}
In any nested sequence $E\in\Enest_{n,m,q}$, online nodes arrive in increasing order of height under \waterfill, and every edge is active.
\end{observation}
\noindent More significantly, the above two properties of an arrival sequence $E$ in fact also give a sufficient condition for the corresponding induced nested sequences to be majorizing: in~\Cref{lem:nestification}, we establish that \emph{the 
relation $\wf(E) \preceq \wf(\wh{E})$ holds when every edge of $E$ is active, and online nodes arrive in non-decreasing order of their height under \waterfill.}

Based on this, our transformation combines a pre-processing routine that transforms any given sequence $E$ to engineer these two properties, followed by taking the induced nested sequence. We specify the full transformation in~\Cref{alg:nest}, and in~\Cref{fig:oper}, we illustrate our procedure on the example from~\Cref{fig:opt_unified}; the latter may be useful to refer to when reading the following.

At a high level, our transformation proceeds as follows:
first, inactive edges are removed, with no effect on the \waterfill
allocation (\Cref{prop:inactive_comp}). 
Second, online nodes are permuted
so that they arrive in non-decreasing order of their height under \waterfill,
with ties broken so that the permutation \emph{reverses} the arrival order of nodes whose \waterfill heights are equal. 
Such a permutation of online nodes' arrival order also turns out to preserve the \waterfill allocation, although establishing this is more involved (see \Cref{prop:neighborhood_pres}).
The key observation is that the \waterfill allocation is preserved only under online node permutations that ensure that each pair of online nodes that have a common (offline) neighbor remains in the same order. At this point, we can use~\Cref{lem:nestification} to get a new sequence $\widetilde{E}$ such that $\wf(E)\preceq\wf{\widetilde(E)}$.

To complete the analysis of \Cref{alg:nest}, we also need to show that $\opt(E) \succeq \opt(\wt{E})$. 
This is not immediate, as in step 1, we remove inactive edges, which could affect the feasibility of the optimal hindsight allocation. 
However, we show that any inactive edge
$(i,t)$ removed in step~1 is replaced with 
the equivalent edge $(i,\sigma(t))$ in step~3 (\Cref{prop:optpreserve}). 
The proof of this fact makes essential use of the arrival-order-reversing tie-breaking rule in step~2. 

In~\Cref{ssec:alg1}, we specify our transformation in~\Cref{alg:nest}, and formalize the above arguments (via~\Cref{prop:inactive_comp,prop:reorder,prop:optpreserve}) and use these to prove~\Cref{lem:nest_wf}. Due to space constraints, we defer proofs of the intermediate propositions to~\Cref{appsec:lemmas}. Given the centrality of~\Cref{lem:nestification} in our approach, we provide a complete proof for it in~\cref{ssec:nestify}.

\begin{figure}[p]
    \centering
    \includegraphics[width=\linewidth]{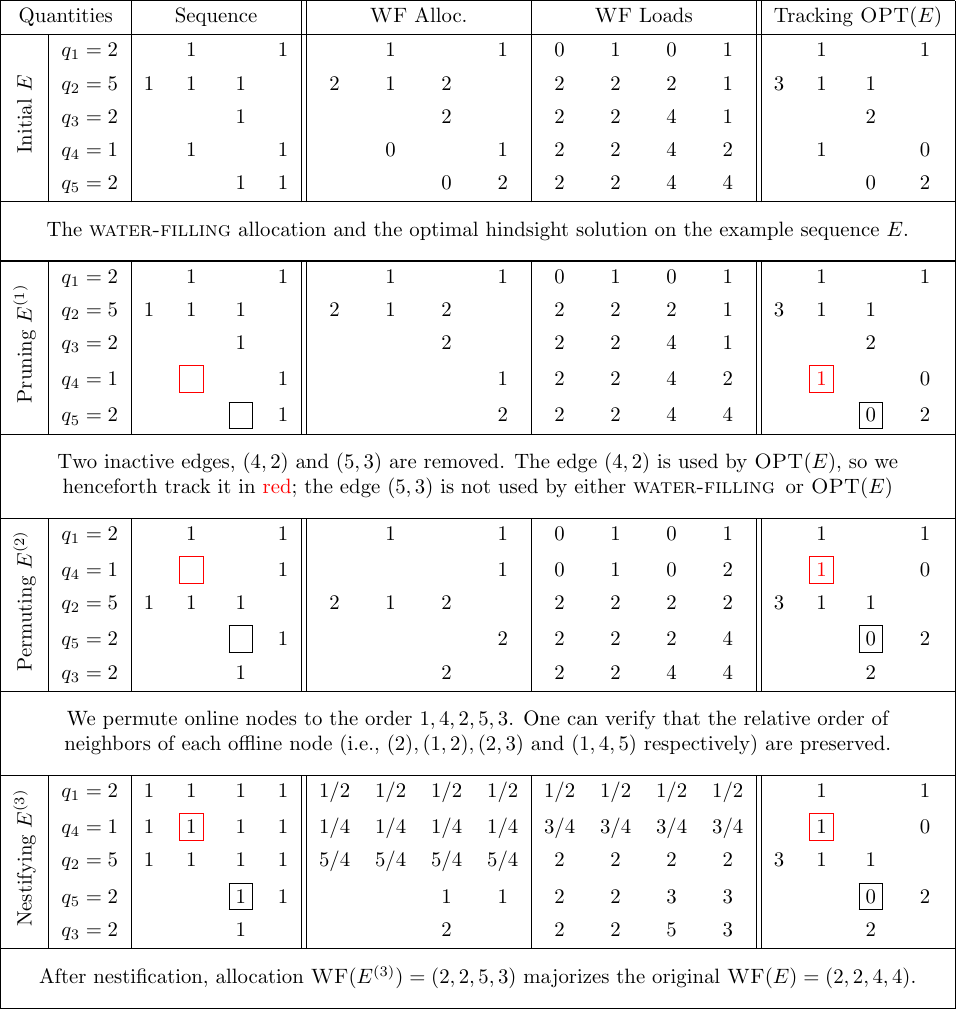}
    \caption{\small\em
        We consider the same allocation sequence from~\Cref{fig:opt_unified}, and demonstrate the transformations (pruning, permuting and nestification) carried out by~\Cref{alg:nest}
        After each step, we update \waterfill allocations; one can verify that the former increases in the majorization preorder (i.e., becomes less balanced). We also track deleted edges in the pruning process via boxes, so one can verify that these are all added back in the end (the red boxes track edges which are in the support of $\opt(E)$). Finally we track the original optimal allocation $\opt(E)$, which becomes infeasible after pruning, but feasible in the final sequence.
    }\label{fig:oper}
\end{figure}

\subsection{Majorizing \waterfill Allocations via Nested Sequences}
\label{ssec:alg1}

We now formally define our procedure for transforming a given request sequence $E$ into a nested sequence $\widetilde{E}$.
\begin{algorithm}[ht]
    \SetAlgoNoLine
    \KwIn{
        A request sequence $E\in\calE_{n,m,q}$.
    }
    \KwOut{
        A nested sequence $\widetilde{E}\in\Enest_{n,m,q}$ with $\wf(E)\preceq\wf(\widetilde{E})$ and $\wf(E)\succeq \wf(\widetilde{E})$.
    }
    \tcc{Step 1: Pruning Inactive Edges}
    Identify inactive edges under the \waterfill allocation $(\mbf{x}_t)_{t\in[m]}$: $I_t \gets \{i\in N_t\mid \mbf{x}_t(i) = 0\}$\;
    Construct sequence $E^{(1)} \gets ((N_t^{(1)}, q_t^{(1)}))_{t\in[m]}\in \calE_{n,m,q}$ with $N_t^{(1)} \gets N_t\setminus I_t$ and $q_t^{(1)} \gets q_t$\;
    
    \BlankLine
    
    \tcc{Step 2: Permuting Online Nodes}
    Identify permutation $\sigma$ satisfying $h_t < h_s\Rightarrow\sigma(t) < \sigma(s)$, and $h_t = h_s$ with $t<s \Rightarrow \sigma(s) < \sigma(t)$\;
    Construct sequence $E^{(2)} \gets ((N_t^{(2)}, q_t^{(2)}))_{t\in[m]}\in \calE_{n,m,q}$ with $N_t^{(2)} \gets N_{\sigma^{-1}(t)}^{(1)}$ and $q_t^{(2)}\gets q_{\sigma^{-1}(t)}^{(1)}$\;
    
    \BlankLine
    
    \tcc{Step 3: Nestification}
    Construct $E^{(3)}\in\Enest_{n,m,q}$ as the nested sequence induced by $E^{(2)}$ (see \Cref{def:induce_nested})\;
    \Return $\widetilde{E} \gets E^{(3)}$\;
    \caption{Majorization-Dispersing Nested Sequence Transformation}
    \label{alg:nest}
\end{algorithm}
The claim in~\Cref{lem:nest_wf} that $\wf(E)\preceq\wf(\widetilde{E})$ depends on establishing two invariance properties under our transformation, and then using our nestification result (\Cref{lem:nestification}). The first states that the \waterfill allocation is preserved by the pruning step (while simultaneously ensuring the resulting sequence has no inactive edges).
\begin{restatable}[Invariance under Inactive Edge Removal]{proposition}{inactivecomp}
\label{prop:inactive_comp}
Consider any sequence $E\in\calE_{n,m,q}$, and sequence $E^{(1)}$ with inactive edges removed as in Step 1 of~\Cref{alg:nest}.
Then $\wf(E)= \wf(E^{(1)})$.   
\end{restatable}
This follows from observing that removing the earliest arriving inactive edge in $E$ does not change any allocation of \waterfill on $E$, and then performing induction over inactive edges. For a detailed proof, see~\Cref{appsec:lemmas}. 

Our second invariance result establishes that the \waterfill allocation in $E^{(1)}$ and $E^{(2)}$ are equivalent (up to permutation). Moreover, by construction, the sequence $E^{(2)}$ also has all its online nodes arriving in non-decreasing order of heights under \waterfill.
\begin{restatable}[Invariance under Neighborhood Preserving Permutations]{proposition}
{nbdpreserve}
\label{prop:neighborhood_pres}
Consider any sequence $E^{(1)}\in\calE_{n,m,q}$, and sequence $E^{(2)}$ obtained by reordering online nodes via permutation $\sigma$ as defined in Step 2 of \Cref{alg:nest}.
Then $\wf(E^{(1)})\sim \wf(E^{(2)})$.  \end{restatable}
The detailed proof of this invariance, which we provide in~\Cref{appsec:lemmas},  is somewhat more subtle. In particular, it depends on establishing (\Cref{prop:reorder}) that the set of permutations under which the \waterfill allocation is invariant corresponds to those which preserve the relative order of online nodes in each $\Gamma_i$ (i.e., in the neighborhood of each offline node). This characterization may be of independent interest.  

At this point, we have a sequence $E^{(2)}$ that has no inactive edges and also has online nodes sorted by height. 
Our next result, which forms the crux of our analysis, establishes that for any such request sequence, the induced nested sequence indeed has the desired majorization relation for \waterfill allocations:
\begin{restatable}[Nestification]{lemma}{nestification}
\label{lem:nestification}
    Let $E^{(2)}\in\calE_{n,m,q}$ be a sequence which under the \waterfill allocation has \underline{no inactive edges} and online nodes arriving in \underline{non-decreasing height order}.
    Then, the nested sequence $E^{(3)}\in\Enest_{n,m,q}$ generated by $E^{(2)}$ satisfies $\wf(E^{(2)}) \preceq \wf(E^{(3)})$.
\end{restatable}

Returning to establishing~\Cref{lem:nest_wf}, we also need to show that the nested sequence $\widetilde{E}$ returned by~\Cref{alg:nest} also preserves the optimal allocation under $E$ (up to permutation). 
\begin{restatable}[\Cref{alg:nest} Preserves $\opt(E)$]{proposition}
{optpreserve}
\label{prop:optpreserve}
Consider any sequence $E\in\calE_{n,m,q}$ with optimal hindsight allocation $(\mbf{x}^*_t)_{t\in[m]}$, and let $\widetilde{E}$ and $\sigma(\cdot)$ respectively be the output of~\Cref{alg:nest}, and the permutation used in Step 2. Then the allocation $(\widetilde{\mbf{x}})_{t\in[m]}$ with $\widetilde{\mbf{x}}_t \define \mbf{x}^*_{\sigma^{-1}(t)}$ is feasible on sequence $\widetilde{E}$
\end{restatable}
\Cref{fig:oper} demonstrates this result visually, by marking and tracking the edges that are removed in pruning, and then showing that they are added back after nestification.
The proof of this result, provided in~\Cref{appsec:lemmas}, makes essential use of the arrival-order-reversing tie-breaking rule in step~2. 
Finally, we can combine these four intermediate results to prove~\Cref{lem:nest_wf} as follows:
\begin{proof}[Proof of~\Cref{lem:nest_wf}]
    Let $E$ be the input to \Cref{alg:nest} which produces variables $E^{(1)}$, $E^{(2)}$, $E^{(3)}$, $\sigma$, and output $\widetilde{E}$.
    \Cref{prop:inactive_comp} implies that $\wf(E) = \wf(E^{(1)})$, and ~\Cref{prop:neighborhood_pres} guarantees that $\wf(E) \sim \wf(E^{(1)})$.
    Further, by \Cref{prop:reorder}, we know that $E^{(2)}$ preserves \waterfill allocations (up to online node reordering) on every edge, and has no inactive edges. 
    Finally, by construction of $\sigma$, we have that $E^{(2)}$ presents online nodes in order of increasing height.
    Thus, applying \Cref{lem:nestification} yields $\wf(E^{(2)}) \preceq \wf(E^{(3)})$.
    This proves the first claim in the lemma:
    \begin{align*}
        \wf(E) = \wf(E^{(1)}) = \wf(E^{(2)}) \preceq \wf(E^{(3)}) = \wf(\widetilde{E})
    \end{align*}
    
    On the other hand, to compare the optimal solutions of $E$ and $\widetilde{E}$, note that by~\Cref{prop:optpreserve}, we have that $\opt(E)$ is a feasible 
    allocation (up to permutation) in $\widetilde{E}$. This proves the second claim, as we have $\opt(\widetilde{E}) \preceq \opt(E)$.
\end{proof}

\subsection{Majorization via Nestification}
\label{ssec:nestify}

Given its centrality in our approach, we conclude this section by providing a proof for~\Cref{lem:nestification}, which establishes sufficient conditions for when induced nested sequences are majorizing. 

\begin{proof}[Proof of~\Cref{lem:nestification}]
    For convenience, we use the notation $E^{(2)} = ((N_t^{(2)}, q_t^{(2)}))_{t\in[m]}$ and $E^{(3)} = ((N_t^{(3)}, q_t^{(3)}))$.
    Also, let $(\bell_t^{(2)})_{t\in [m]}$ and $(\bell_t^{(3)})_{t\in[m]}$ be the intermediate allocations produced by \waterfill on $E^{(2)}$ and $E^{(3)}$, respectively.
    For notational convenience, we define $\bell_0^{(2)} \define \bell_0^{(3)} \define \vec{0}$, and denote the \emph{latest arriving} online node in the neighborhood of offline node $i$ on sequence $E^{(2)}$ as $\mu_i \define \max(\Gamma_i)$ (with $\mu_0 \define 0$).
    As an example, for the request sequence $E^{(2)}$ in~\Cref{fig:oper}, we have $\mu=(3,3,5,4)$.
    Note that $\mu_i$ is also the latest arriving online neighbor of $i\in[n]$ on sequence $E^{(3)}$ by construction of the induced nested sequence.

    Without loss of generality, we first relabel the offline nodes so that their latest arriving neighbors are in increasing order: $\mu_1 \le \dots\le\mu_n$.
    This labeling assumption implies that $\bell_m^{(2)}(1) \le\dots\le \bell_m^{(2)}(n)$ (since by assumption, online nodes in $E^{(2)}$ arrive in order of non-decreasing heights) and $\bell_{m}^{(3)}(1) \le\dots\le \bell_m^{(3)}(n)$ (by~\Cref{obs:wf_on_nest}, as $E^{(3)}$ is a nested request sequence).
    We recall also that for any vector $x\in\R^n$ and $i\in n$, we use the shorthand notation $x([i])=\sum_{t=1}^ix(t)$.
    
    Now to prove majorization, we use the increasing-cumulative-sum characterization (\Cref{def:maj}).
    Specifically, since $\wf(E^{(2)})([n]) = \wf(E^{(3)})([n])$, showing majorization is equivalent to:
    \begin{align*}
        \wf(E^{(3)})^{\uparrow}([i]) &\le \wf(E^{(2)})^{\uparrow}([i])
        ,\quad\,\forall i\in[n] && (\text{by~\Cref{def:maj}})\\
        \Leftrightarrow \bell_m^{(3)}([i]) &\le \bell_m^{(2)}([i])
        ,\qquad\qquad\forall i\in[n] 
        &&(\text{since load vectors are sorted})\\
        \Leftrightarrow \bell_{\mu_i}^{(3)}([i]) &\le \bell_{\mu_i}^{(2)}([i])
        ,\qquad\qquad\forall i\in[n]
        && (\text{since } \mu_i\text{ is the last online neighbor of }i)
    \end{align*}

    Define $\mathcal{I}\define \{i\in[n-1]:\mu_i < \mu_{i+1}\}\cup\{0\}$ to denote the set of offline nodes whose latest arriving online neighbor differs from that of the subsequent offline node (again to help visualize this, in~\Cref{fig:oper} we can relabel the offline nodes in the order $(1,2,4,3)$ so that the last arriving node sequence is $\mu=(3,3,4,5)$; now the set $\mathcal{I}$ comprises of nodes $\{2,4\}$). We next argue that it is sufficient to establish $\bell_{\mu_i}^{(3)}([i]) \leq \bell_{\mu_i}^{(2)}([i])\,\forall i\in \mathcal{I}$. To see this, consider extending the functions $i \mapsto \bell^{(2)}_{\mu_i}([i])$ and $i \mapsto \bell^{(3)}_{\mu_i}([i])$ to  piecewise linear functions $h^{(2)}, h^{(3)}$  (respectively) on $[0,n]$. Since the functions are equal at $0$ and $n$, it is therefore sufficient to establish $\bell_{\mu_i}^{(3)}([i]) \leq \bell_{\mu_i}^{(2)}([i])$ at any $i\in [n-1]$ where the slope of the piecewise linear extensions change. However for $i\in[n-1]\setminus\mathcal{I}$, we have:  
    \begin{align}
        \bell_{\mu_{i+1}}^{(2)}([i+1]) - \bell_{\mu_{i}}^{(2)}([i])
        &= \bell_{\mu_i}^{(2)}([i+1]) - \bell_{\mu_i}^{(2)}([i])
        &&(\text{since }\mu_i = \mu_{i+1}) \nonumber \\
        &= \bell_{\mu_i}^{(2)}(i+1) 
        = \bell_{\mu_i}^{(2)}(i)
        \label{eq:worst_case_break1} \\
        &= \bell_{\mu_i}^{(2)}([i]) - \bell_{\mu_{i}}^{(2)}([i-1]) \nonumber \\
        &= \bell_{\mu_i}^{(2)}([i]) - \bell_{\mu_{i-1}}^{(2)}([i-1])
        \label{eq:worst_case_break2}
    \end{align}
    Line \ref{eq:worst_case_break1} follows from the equalized loads property of \waterfill (\Cref{obs:wf_properties}) since $i$ and $i+1$ are both neighbors of $\mu_i = \mu_{i+1}$ with active edges.
    Line \ref{eq:worst_case_break2} holds because offline nodes in $[i-1]$ have no neighbors after $\mu_{i-1}$.
    Note that the same analysis applies to the piecewise extension of $\bell^{(3)}(\cdot)$.
    
    We have shown that the slopes of the piecewise linear functions $h^{(2)},h^{(3)}$
    only change at points $i \in \mathcal{I}=\{i\in[n-1]:\mu_i < \mu_{i+1}\}\cup\{0\}$, 
    so we now turn to establishing $h^{(3)}(i) \le h^{(2)}(i)$ for $i \in \mathcal{I}$.
    Next, note that for $i\in\mathcal{I}$, we have $[i] = [n]\setminus N_{\mu_{i+1}}^{(3)}$; this follows from the fact that $E^{(3)}$ is a nested sequence, and that $\mu_{i+1}>\mu_i$. Consequently, for all $i\in\mathcal{I}$, we have:
    \begin{align*}
        \bell_{\mu_i}^{(3)}([i]) \le \bell_{\mu_i}^{(2)}([i])        &\Leftrightarrow 
        \bell_{\mu_i}^{(3)}([n]\setminus N_{\mu_{i+1}}^{(3)}) \le \bell_{\mu_i}^{(2)}([n]\setminus N_{\mu_{i+1}}^{(3)})\\
        &\Leftrightarrow \bell_{\mu_i}^{(3)}(N_{\mu_{i+1}}^{(3)}) \ge \bell_{\mu_i}^{(2)}(N_{\mu_{i+1}}^{(3)})
        \qquad\quad(\text{since }
        \bell_{\mu_i}^{(3)}([n]) = \bell_{\mu_i}^{(2)}([n]))
    \end{align*}
    We now show using induction that this last inequality holds for all $i\in \calI$.
    The base case $i=0$ is true by definition, as $\mu_0 = 0$ and $\bell_{0}^{( 2)}=\bell_{0}^{(3)}=\vec{0}$. 
    To simplify notation, we omit superscript $(3)$ from all neighborhoods for the rest of the proof: $N_t \gets N_t^{(3)}$.

        Now, for any $j\in\calI$, assume all inequalities with $i\le j-1$ for $i\in\calI$ hold.
        Take $k\in\calI$ to be the largest index satisfying $k\le j-1$.
        By the definition of $\calI$, we know that $\mu_j < \mu_{j+1}$.
        Further, the neighborhoods of online nodes $\mu_{k} + 1,\dots,\mu_{j}$ in both $E^{(2)}$ and $E^{(3)}$ are all subsets of $N_{\mu_{j}}$, and the total additional load introduced by these online nodes $\sum_{t=\mu_{k}+1}^{\mu_{j}} q_t$.
        The inductive assumption, together with the equation $N_{\mu_{k+1}} = N_{\mu_j}$, thus implies:
        \begin{align}\label{eq:nested_ineq1}
            \bell_{\mu_{j}}^{(2)}(N_{\mu_j}) = \sum_{t=\mu_{k}+1}^{\mu_{j}} q_t + \bell_{\mu_{k}}^{(2)}(N_{\mu_{k+1}}) \le \sum_{t=\mu_{k}+1}^{\mu_{j}} q_t + \bell_{\mu_{k}}^{(3)}(N_{\mu_{k+1}}) = \bell_{\mu_{j}}^{(3)}(N_{\mu_j})
        \end{align}
    The nested structure of $E^{(3)}$ guarantees that the \waterfill load of each offline node in $N_{\mu_j}$ is equal after the arrival of online node $\mu_j$ (by \Cref{obs:wf_properties,obs:wf_on_nest}).
    Formally, for all $i\in N_{\mu_j}$ we have $\bell_{\mu_j}^{(3)}(i) = \bell_{\mu_j}^{(3)}(N_{\mu_j}) / |N_{\mu_j}|$.
    Taking the sum of loads on $i\in N_{\mu_j}$ we find:
    \begin{align}\label{eq:nested_ineq2}
        \bell_{\mu_j}^{(3)}(N_{\mu_{j+1}}) = \left(\frac{|N_{\mu_{j+1}}|}{|N_{\mu_j}|}\right) \bell_{\mu_j}^{(3)}(N_{\mu_j})
    \end{align}
    Consider the load of offline nodes in $N_{\mu_j}$ after the arrival of online node $\mu_j$ for sequence $E^{(2)}$.
    Notice that offline nodes $i\in N_{\mu_j}\setminus N_{\mu_j + 1}$ have an active edge to online node $\mu_j$.
    Moreover, online nodes in $E^{(2)}$ arrive in order of increasing height.
    Thus, all $i\in N_{\mu_j}\setminus N_{\mu_{j+1}}$ satisfies $\bell_{\mu_j}^{(2)}(i) \in \arg\max_{k\in N_{\mu_j}} \bell_{\mu_j}^{(2)}(k)$.
    The total load of offline nodes in $N_{\mu_j}\setminus N_{\mu_{j+1}}$ after the arrival of online node $\mu_j$ on sequence $E^{(2)}$ is:
    \begin{align}\label{eq:nested_ineq3}
        \bell_{\mu_j}^{(2)}(N_{\mu_j}\setminus N_{\mu_{j+1}}) \ge \left(\frac{|N_{\mu_j}\setminus N_{\mu_{j+1}}|}{|N_{\mu_j}|}\right)\bell_{\mu_j}^{(2)}(N_{\mu_j})
    \end{align}
    Combining \Cref{eq:nested_ineq1,eq:nested_ineq2,eq:nested_ineq3} finishes the induction as follows:
    \begin{align*}
        \bell_{\mu_j}^{(2)}(N_{\mu_{j+1}})
        = \bell_{\mu_j}^{(2)}(N_{\mu_j}) - \bell_{\mu_j}^{(2)}(N_{\mu_j}\setminus N_{\mu_{j+1}}) 
        &\le \bell_{\mu_j}^{(2)}(N_{\mu_j}) - \left(\frac{|N_{\mu_j}\setminus N_{\mu_{j+1}}|}{|N_{\mu_j}|}\right)\bell_{\mu_j}^{(2)}(N_{\mu_j})
        &&\text{(from Eq.~\eqref{eq:nested_ineq3})} \\
        &= \left(\frac{|N_{\mu_{j+1}}|}{|N_{\mu_j}|}\right)\bell_{\mu_j}^{(2)}(N_{\mu_j})  \\
        &\le \left(\frac{|N_{\mu_{j+1}}|}{|N_{\mu_j}|}\right)\bell_{\mu_j}^{(3)}(N_{\mu_j})
        &&\text{(from Eq. \eqref{eq:nested_ineq1})} \\
        &= \bell_{\mu_j}^{(3)}(N_{\mu_{j+1}})
        &&\text{(from Eq.~\eqref{eq:nested_ineq2})}
    \end{align*}
    Thus, we have $\wf(E^{(2)})\preceq \wf(E^{(3)})$.
\end{proof}

\section{Arbitrary Allocation Policies and Nested Sequences}\label{sec:wf_vs_policy}

In~\Cref{sec:nest}, we establish that for \waterfill, the worst-case sequences (in terms of majorization dispersion) are nested. Consequently, to compare to other policies, we can consider only such nested sequences. 
For any given policy $\calA$ and nested sequence $\wt{E}$, we now describe how to transform it into sequence $E'$ by \emph{relabeling the offline nodes} (and so preserving $\opt(\wt{E})$ and $\wf(\wt{E})$), such that $\calA$ now majorizes \waterfill.
At a high level: iteratively at each time step, we retain only those offline nodes that accumulate larger expected loads under $\calA$, and remove nodes with smaller expected loads.
This ensures nodes that are ``harder'' for $\calA$ to serve remain active longer.

\begin{algorithm}[ht]
    \SetAlgoNoLine
    \KwIn{A nested sequence $\wt{E}\in\Enest_{n,m,q}$ and a (possibly randomized) allocation policy $\calA$.}
    \KwOut{A nested sequence $E'\in\Enest_{n,m,q}$ such that $\wf(\wt{E})\preceq \E[\calA(E')]$ and $\opt(\wt{E})\sim\opt(E')$.}

    $\forall\,i\in[n]$, set $\wt{\mu}_i \gets \max(\Gamma_i(\wt{E})\cup\{0\})$ (last online neighbor); $\forall\,t\in[m]$ define
    $\phi_t \gets |\{i : \wt{\mu}_i=t\}|$\;
    Initialize $S_1 \gets [n]$\;

    \For{$t=1,\dots,m$}{
        Set $N_t' \gets S_t$ and $q_t' \gets q_t$\;
        Compute allocations:\qquad $\mbf{x}_t^{\calA} \gets \calA\big((N_1',q_1',\mbf{x}_1^{\calA}),\dots,(N_{t-1}',q_{t-1}',\mbf{x}_{t-1}^{\calA}),(N_t',q_t')\big)$\\
        \,\,\,and (expected) load:\qquad $\E[\bell_t^{\calA}] \gets \E[\sum_{s=1}^t \mbf{x}_s^{\calA}]$\;
        Remove from $S_t$ the $\phi_t$ nodes with smallest values in vector $\E[\bell_t^{\calA}]$\;
        Set $S_{t+1}\gets S_t$\;
    }
    \Return sequence $E' \gets ((N_t',q_t'))_{t\in[m]}$\;
    \caption{\waterfill vs. Allocation Policies}
    \label{alg:wf_vs_policy}
\end{algorithm}
\wfvspolicy*
\begin{proof}
 Since $\wt{E},E'$ only differ in offline node labels $\Rightarrow \wf(\wt{E})\sim\wf(E')$ and $\opt(\wt{E})\sim\opt(E')$.
 It remains for us to show $\wf(E')\preceq \E[\calA(E')]$.
 All quantities below are defined with respect to $E'$.

Let $(\bell_t^{\wf})_{t\in[m]}$ denote intermediate \waterfill loads, and for each offline node $i$, let $\mu_i$ be its last online neighbor.
We relabel offline nodes so that
$
\E[\bell_m^{\calA}(1)] \leq \E[\bell_m^{\calA}(2)] \leq \cdots \leq \E[\bell_m^{\calA}(n)] .
$
\Cref{alg:wf_vs_policy} now removes nodes in non-decreasing order of expected loads, so
$\mu_1 \leq \mu_2 \leq \cdots \leq \mu_n$.
Since $E'$ is nested, this implies
$\bell_m^{\wf}(1)\leq\bell_m^{\wf}(2)\leq \cdots\leq\bell_m^{\wf}(n)$.

Define $\bell_m^{\wf}(0)=\bell_m^{\calA}(0)=0$.
By the cumulative-sum characterization of majorization (\Cref{def:maj}), it suffices to show $\bell_m^{\wf}([i]) \ge \E[\bell_m^{\calA}([i])] \,
\forall\, i\in[n]\cup\{0\}$.
We do so by induction on $i$. Note the case $i=0$ holds by definition.

Now fix $i\ge1$ and assume the claim holds for $j\leq i-1$.
Since nodes in $[i]$ receive no allocation after online node $\mu_i$, we compare loads at $\mu_i$.
Under \waterfill, the total load is $\sum_{t=1}^{\mu_i} q_t$, of which
$\bell_{\mu_i}^{\wf}([i-1])$ is assigned to nodes $[i-1]$.
Since $E'$ is a nested sequence, \Cref{obs:wf_properties,obs:wf_on_nest} imply \waterfill has equal load across the remaining nodes $k\ge i$, and so:
    \begin{align}\label{eq:wf_vs_policy1}
        \bell_{\mu_i}^{\wf}(i)
        = \left(\frac{1}{n-i+1}\right)\left(\sum_{t=1}^{\mu_i} q_t - \bell_{\mu_i}^{\wf}([i-1])\right)
    \end{align}
For policy $\calA$, the total expected load after online node $\mu_i$ is also $\sum_{t=1}^{\mu_i}q_t$, and the expected load allocated to nodes in $[i-1]$ is $\E[\bell_{\mu_i}^{\calA}([i-1])]$.
Since offline node $i$ has the smallest expected intermediate load of nodes in $[n]\setminus [i-1]$ after the arrival of online node $\mu_i$, we have:
    \begin{align}\label{eq:wf_vs_policy2}
        \E[\bell_{\mu_i}^{\calA}(i)]
        \le \left(\frac{1}{n-i+1}\right)\left(\sum_{t=1}^{\mu_i} q_t - \E[\bell_{\mu_i}^{\calA}([i-1])]\right)
    \end{align}
    Since $(n-i+1)\bell_{\mu_i}^{\wf}([i]) = (n-i+1) \bell_{\mu_i}^{\wf}([i-1])+(n-i+1)\bell_{\mu_i}^{\wf}(i)$, applying~\Cref{eq:wf_vs_policy1} we get
    $(n-i+1)\bell_{\mu_i}^{\wf}([i]) =  (n-i)\bell_{\mu_i}^{\wf}([i-1]) + \left(\sum_{t=1}^{\mu_i} q_t\right)$. Similarly, by applying \Cref{eq:wf_vs_policy2}, we get $(n-i+1)\E[\bell_{\mu_i}^{\calA}([i])] \leq  (n-i)\E[\bell_{\mu_i}^{\calA}([i-1])] + \left(\sum_{t=1}^{\mu_i} q_t\right)$.
    Finally, the inductive hypothesis gives:
    \begin{align*}
    \bell_m^{\wf}([i])
=
\bell_{\mu_i}^{\wf}([i])
\ge
\E[\bell_{\mu_i}^{\calA}([i])]
=
\E[\bell_m^{\calA}([i])].
\end{align*}
Thus the claim holds for all $i$, proving $\wf(E')\preceq \E[\calA(E')]$.
\end{proof}

\section{Adversarial Sequences and Regret Implications}
\label{sec:uppertriangle}

\Cref{thm:wf_min} shows that nested sequences include the adversary's best response to a principal using \waterfill.
However, not all nested sequences are bad; for example, the nested sequence $E\in\calE_{n,m,q}$ with $N_t \gets [n]$
and $q_t \gets \frac{q}{m}$ has \waterfill allocation $\wf(E) = \frac{q}{n}\vec{1}$, which is minimal in the majorization preorder.
In this section, we provide a more precise characterization of the adversary's best strategy.
As we mention before, our aim here is to show that for any given number of offline nodes $n$ and total quantity $q$, it is sufficient for an adversary to only consider \emph{$n\times n$ complete upper-triangular} request sequences (i.e., with $m=n$ online nodes and $N_t=\{t,t+1,\ldots,n\}$) and \emph{non-decreasing quantities} (i.e., $q_1\leq q_2\leq\ldots\leq q_n$ and $\sum_{t=1}^nq_t=q$). To do so, we transform any given nested sequence into a complete upper-triangular sequence via the following procedure.
\begin{algorithm}[ht]
    \SetAlgoNoLine
    \KwIn{
        A nested sequence $\widetilde{E}\in\Enest_{n,m,q}$.
    }
    \KwOut{
        A nested sequence $E'\in\Enest_{n,n,q}$ satisfying \Cref{thm:param}.
    }
    \For{$t=1,\dots,n$}{
        Define $N_t' \gets \{t,t+1,\dots, n\}$ and $q_t' \gets \left(\opt(\widetilde{E})\right)^{\uparrow}(t)$\;
    }
    \Return sequence $E'\gets ((N_t', q_t'))_{t\in[n]}$\;
    \caption{Constructing Worst-Case Sequences for \waterfill}\label{alg:param_nest}
\end{algorithm}%

\param*
Before proving this, we need some additional definitions and results.
First, the linear map $H$, defined by $H(i,j)=\mathbf{1}\{i\ge j\}/(n-j+1)$, captures the behavior of \waterfill on any nested request sequence $E=((N_t,q_t))_{t\in[m]}\in\Enest_{n,m,q}$ as follows: \waterfill distributes $q_t$ uniformly over $N_t$, equalizing intermediate loads.
Thus, an offline node $i$ whose last neighbor is $\mu_i$ receives
$
\sum_{t=1}^{\mu_i}\frac{q_t}{|N_t|}.
$
This is a weighted harmonic sum, where the term with denominator $n-i+1$ aggregates the mass of online nodes with exactly $n-i+1$  neighbors.
Hence, the \waterfill load vector is linear in these weights. When $E$ is complete upper-triangular, then it is easy to observe that $\wf(E)=Hq$; for general nested sequences we have the following (proof deferred to~\Cref{appsec:uppertriangle}).
\begin{restatable}[\waterfill Transform $H$]{observation}{hwf}\label{obs:h_wf}
    Fix a nested sequence $E=((N_t, q_t))_{t\in[m]} \in \calE_{n,m,q}$ and assume that the latest online neighbor $\mu_i = \max(\Gamma_i(E))$ of offline nodes satisfy $\mu_1\le \dots\le\mu_n$.
    Define the vector $\mbf{z}\in \R_{\ge 0}^n$ with $\mbf{z}(i) = \sum_{t\in \Gamma_i(E) \setminus \Gamma_{i-1}(E)} q_t$, where $\Gamma_0(E) \define \emptyset$.
    The \waterfill allocation satisfies $\wf(E) = H\mbf{z}$.
\end{restatable}
The transformation $H$ has an additional property that it preserves majorization order of a pair of vectors $\mbf{x},\mbf{y}$ that satisfy a partial cumulative sum condition expressed below.
\begin{restatable}[$H$ is a Majorizing Transform]{proposition}{hmajpres}
\label{obs:h_maj}
    Fix $\mbf{x},\mbf{y}\in\R_{\ge 0}^n$ with $\mbf{x}([n]) = \mbf{y}([n]) = q$. 
    If $\mbf{x}([k]) \le \mbf{y}([k])$ for all $k\in[n]$, then $H \mbf{x} \succeq H \mbf{y}$.
\end{restatable}
We prove this in~\Cref{appsec:uppertriangle}, but essentially, the map $H$ redistributes mass from index $i$ uniformly over $\{i,\dots,n\}$, so mass on smaller indices is spread more widely than mass on larger ones. 
Hence, vectors with smaller prefix sums yield less equitable images under $H$.
Notably, this holds without requiring either vector to be sorted. Now we can turn to proving~\Cref{thm:param}.

\begin{proof}[Proof of~\Cref{thm:param}]
First, we observe that by construction of $E'$, \Cref{item:construct,item:order} hold immediately.
In particular, $q'$ is the sorted vector $\opt(\wt{E})$.

Denote $\wt{E}=((\wt{N}_t,\wt{q}_t))_{t\in[m]}$ and $E'=((N_t',q_t'))_{t\in[m]}$.
Without loss of generality, we relabel offline nodes so that $\wf(\widetilde{E})(1) \le \dots \le \wf(\widetilde{E})(n)$, which implies $\opt(\widetilde{E})(1)\le\dots\le\opt(\widetilde{E})(n)$ since $\widetilde{E}$ is nested.
Let $(\wt{\mbf{x}}_t)_{t\in[m]}$ and $\wt{\mbf{z}}$ be the \waterfill allocation and associated vector from~\Cref{obs:h_wf} on $\wt{E}$, and define $(\mbf{x}_t')_{t\in[m]}$ and $\mbf{z}'$ analogously for $E'$.
We define an additional allocation $(\mbf{x}_t^*)_{t\in[n]}$ with $\mbf{x}_t^*(i) \define q_t'\cdot \1{i=t}$ inducing cumulative load $\bell^* \define \sum_{t\in[m]}\mbf{x}_t^*\sim \opt(\widetilde{E})$. Note that $\mbf{x}_t^*$ is feasible on $E'$ by definition.

\paragraph{Optimality.}
We show $\opt(E')\sim\opt(\wt{E})$ by proving that $\bell^*$ is majorization-minimal on $E'$.
Since $\Gamma_i(E')=[i]$, the $j$ smallest offline nodes have neighbors only among the first $j$ online nodes, whose total quantity is $\sum_{i\le j} q_i'$.
Thus any feasible allocation satisfies $\bell([j])\le \sum_{i\le j} q_i'$.
The allocation $\mbf{x}^*$ attains equality for all $j$, so $\bell^*$ is majorization-minimal and hence optimal.

\paragraph{\waterfill allocations.}
To show $\wf(\wt{E})\preceq\wf(E')$, we compare $\wt{\mbf{z}}$ and $\mbf{z}'$.
For $E'$, using nesting,
\begin{align*}
\mbf{z}'([j])
&= \sum_{i\le j}\sum_{t\in\Gamma_i(E')\setminus\Gamma_{i-1}(E')} q_t'
= \sum_{t\in\bigcup_{i\le j}\Gamma_i(E')} q_t'
 = \sum_{t\le j} q_t'
 = \opt(\wt{E})([j]).
\end{align*}
Similarly,
$
\wt{\mbf{z}}([j])
= \sum_{t\in\bigcup_{i\le j}\Gamma_i(\wt{E})} \wt{q}_t .
$
This quantity is the total supply available to offline nodes $[j]$ in $\wt{E}$ and must be at least $\opt(\wt{E})([j])$ for feasibility.
Hence $\wt{\mbf{z}}([j])\ge \mbf{z}'([j])$ for all $j$, with equality at $j=n$.
Applying \Cref{obs:h_wf} and \Cref{obs:h_maj},
\[
\wf(\wt{E}) = H\wt{\mbf{z}} \preceq H\mbf{z}' = \wf(E'),
\]
completing the proof.
    
\end{proof}

\newpage 

\bibliographystyle{ACM-Reference-Format}
\bibliography{references}

\appendix

\section{Adaptive Adversaries}\label{app:sec:adapt_adv}
In the main body of our work, we focused on an oblivious adversary that designs the entirety of the request sequence $E\in\calE_{n,m,q}$ without knowledge of the allocations produced by the algorithm designer's policy $\calA$.
Against such adversaries, \Cref{cor:reg_det} and \Cref{thm:no_opt_rand} show there is an unavoidable separation between the classes of objections for which \waterfill is minimax optimal that depends on whether or not the algorithm designer is allowed randomness.
This separation arises because a randomized policy can, with non-zero probability, guess an equitable allocation on any fixed request sequence.

In this section, we turn our attention to a stronger adaptive adversary who can dynamically design the request sequence depending on the play of the algorithm designer's policy.
An adaptive adversary, when the algorithm designer deploys a deterministic policy such as \waterfill, is no stronger than an oblivious one.
However, adaptive adversaries can prevent randomized policies from guessing equitable allocations.
As a result, there is no separation between the deterministic and randomized minimax guarantees of \waterfill against adaptive adversaries.

\subsection{Adaptive Adversary Formalization}
We use a similar game setting to the one described in \Cref{subsec:setting}.
However, instead of being a static request sequence, the adversary's policies are maps from a history of past requests (online nodes) and the algorithm designer's allocations on those nodes to the next online node in the sequence.
Up to time $t$, the adversary generates online nodes $((N_s, q_s))_{s\in[t-1]}$ and the algorithm designer plays feasible allocations $(\mbf{x})_{s\in[t-1]}$ on those nodes.
On round $t$, an adaptive adversary, with knowledge of $((N_s, q_s))_{s\in[t-1]}$ and $(\mbf{x})_{s\in[t-1]}$, designs the characteristics of $t$th online node $(N_t, q_t)$.
We overload notation and use $E$ to indicate both static request sequences and a map designed by an adaptive adversary.
Similarly, the set of adaptive request sequences on which the final trajectory of requests has $n$ offline nodes, $m$ online nodes, and $q$ total quantity overloads $\calE_{n,m,q}$.

Similar to the oblivious setting, a game instance is defined by the number of offline nodes $n\in\N$, the number of online nodes $m\in\N$, the total quantity of all online nodes $q\in\R_{>0}$, and a comparison factor $\alpha \in\R_{> 0}$.
When the algorithm designer uses policy $\calA$ and the adversary plays an adaptive request sequence $E$, a possibly randomized static request sequence $((N_t, q_t))_{t\in[m]}$ and allocations $(\mbf{x}_t)_{t\in[m]}$ are generated.
The allocation (load vector) generated by $\calA$ on this trajectory is $\calA(E)\define \sum_{t} \mbf{x}_t$
We use $\opt(\calA, E)$ to denote the possibly randomized majorization minimal load vector on $((N_t, q_t))_{t\in[m]}$.
Such a load vector exists by \Cref{fact:opt}.
The cost incurred by the algorithm designer trying to maximize equity measuring objective (Schur-concave function) $f$ is:
\begin{align*}
    c_{\alpha,f}^{\max}(\calA, E)
    = \E[\alpha \cdot f(\opt(\calA, E)) - f(\calA(E))]
\end{align*}
Similar to the oblivious setting, $\alpha$-regret of $\calA$ is defined as the maximum cost incurred by the algorithm designer against an adaptive adversary:
\begin{align*}
    \reg_{(n,m,q),\alpha, f}^{\max}(\calA)
    = \sup_{\text{adaptive }E\in \calE_{n,m,q}} c_{\alpha, f}^{\max}(\calA, E)
\end{align*}
The cost and regret $\reg_{(n,m,q),\alpha,g}^{\min}$ for minimization of Schur-convex $g$ are analogously defined.

\subsection{Adaptive Adversary Results}
First, we establish that adaptivity provides the adversary no more power against deterministic allocation policies than oblivious strategies.
In such cases, the deterministic policy $\calA$ and adaptive request sequence $E$ generate deterministic requests $((N_t, q_t))_{t\in[m]}$ and allocations $(\mbf{x}_t)_{t\in[m]}$.
When the adversary moves first with access to deterministic policy $\calA$, they may simulate the policy's behavior for different request sequences, so adaptivity provides no benefit.

\begin{fact}[Adaptivity and Determinism]\label{app:fact:adapt_det}
    Fix an adaptive request sequence $E\in \calE_{n,m,q}$ and allocation policy $\calA$, which generates deterministic request sequences $((N_t, q_t))_{t\in[m]}$ and allocations $(\mbf{x}_t)_{t\in[m]}$.
    On the static and oblivious request sequence $((N_t, q_t))_{t\in[m]}\in\calE_{n,m,q}$, policy $\calA$ produces the same allocations $(\mbf{x}_t)_{t\in[m]}$.
\end{fact}

As a result of the above fact, our analysis involving the deterministic \waterfill policy and oblivious adversaries generalizes to the adaptive setting.
The only argument that needs to be extended is \Cref{lem:wf_vs_policy}, which has implications for the adaptive versions of \Cref{thm:wf_min} as well as \Cref{cor:reg_det,cor:reg_rand}.
The adaptive variant of these statements follows:

\begin{restatable}[Adaptive Majorization Minimality]{theorem}{wfminadap}\label{app:thm:wf_min_adapt}
    Fix parameters $n,m\in\N$ and $q\in\R_{>0}$.
    Given any (possibly randomized) allocation policy $\calA$ and any oblivious request sequence $E\in\calE_{n,m,q}$, there is an adaptive request sequence $E'$ for which $\wf(E) \preceq \wf(\calA(E'))$ and $\opt(\wf, E) \succeq \opt(\calA, E')$ with probability $1$.
\end{restatable}

\Cref{app:thm:wf_min_adapt} is a composition of \Cref{lem:nest_wf} and an adversarial variant of \Cref{lem:wf_vs_policy} stated below.

\begin{restatable}[Policy Deviation]{lemma}{wfvspolicyadapt}\label{app:lem:wf_vs_polict_adapt}
    Given any policy $\calA$ and an oblivious nested sequence $\wt{E}\in\Enest_{n,m,q}$ as input, \Cref{app:alg:wf_vs_policy_adapt} outputs an adaptive request sequence $E'\in\calE_{n,m,q}$ that guarantees $\wf(\wt{E}) \preceq \calA(E')$ and $\opt(\wf, \wt{E}) \sim \opt(\calA, E')$ with probability $1$. 
\end{restatable}

\Cref{app:thm:wf_min_adapt} implies that $\alpha$-regret for any Schur-monotone objective is minimized by \waterfill, even when considering randomized polices.
This stands in contrast to the separation in the oblivious setting, where competing with randomized algorithms required constraining the class of objectives from Schur-monotone objectives to symmetric and concave/convex ones.

\begin{restatable}[Adaptive Minimax Optimality]{corollary}{wfadapt}\label{app:cor:wf_adapt}
    Fix $n,m,q,\alpha$, and a Schur-concave objective $f:\R_{\ge 0}^n\to\R_{\ge 0}$ (or Schur-convex objective $g:\R_{\ge 0}^n \to \R_{\ge 0}$). The $\alpha$-regret against an adaptive adversary of $\wf$ in comparison to that of any (possibly randomized) allocation policy $\calA$ satisfies:
    \begin{align*}
        \reg_{(n,m,q),\alpha,f}^{\max}(\wf) \le \reg_{(n,m,q),\alpha,f}^{\max}(\calA)
        \quad\quad\text{or}\quad\quad
        \reg_{(n,m,q),\alpha,g}^{\min}(\wf) \le \reg_{(n,m,q),\alpha,g}^{\min}(\calA)
    \end{align*}
\end{restatable}

\subsection{Policy Deviation Against Adaptive Adversaries}
The adaptive adversarial request sequence \Cref{app:alg:wf_vs_policy_adapt} yielding \Cref{app:lem:wf_vs_polict_adapt} closely resembles \Cref{alg:wf_vs_policy}.
Given a seed oblivious nested sequence $\wt{E}\in\calE_{n,m,q}$, the adaptive policy relabels offline nodes.

\begin{algorithm}[ht]
    \SetAlgoNoLine
    \KwIn{
        An oblivious nested sequence $\wt{E}\in\Enest_{n,m,q}$ and a (possibly randomized) allocation policy $\calA$.
    }
    \KwOut{
        An adaptive sequence $E'\in\Enest_{n,m,q}$ with $\Pr[\wf(\widetilde{E})\preceq\calA(E')\text{ and }\opt(\widetilde{E})\sim \opt(E')]=1$.
    }
    Initialize the last neighbor of $i\in [n]$ on $\widetilde{E}$ as $\widetilde{\mu}_i \gets\max(\Gamma_i(\widetilde{E})\cup\{0\})$ and $\phi_t = |\{i\mid \widetilde{\mu}_i = t\}|$\;
    Define $S_1 \gets [n]$\;
    \For{$t=1,\dots,m$}{
        Instantiate $N_t'\gets S_t$ and $q_t' \gets q_t$\;
        Play online node $(N_t',q_t')$ \;
        Observe $\mbf{x}_t^\calA$ played by the algorithm\;
        Compute the randomized load vector adaptively $\bell_t^{\calA}\gets \sum_{s=1}^{t}\mbf{x}_t^{\calA}$\;
        From $S_t$ remove $\phi_t$ offline nodes with the smallest value in vector $\bell_t^{\calA}$\;
    }
    \caption{\waterfill vs. Allocation Policies}\label{app:alg:wf_vs_policy_adapt}
\end{algorithm}

\wfvspolicyadapt*

\begin{proof}
    The possibly randomized sequence $(N_t', q_t')$ is isomorphic to $E'$ under offline node labeling with probability $1$.
    The \waterfill allocation and the optimal hindsight solution will be equivalent to those on $\wt{E}$ up to permutation: $\wf(\wt{E}) \sim \wf(E')$ and $\opt(\wt{E})\sim \opt(\calA, E')$.
    The proof is completed by showing that $\wf(E') \preceq \calA(E')$ with probability $1$.
    For the remainder of the proof, we define everything with respect to $E'$.

    Fix a realization of the request sequence $((N_t', q_t'))_{t\in[m]}$ and the intermediate load vectors that $\calA$ produces $(\bell_t^{\calA})_{t\in[m]}$.
    Let $(\bell_t^{\wf})_{t\in[m]}$ be the intermediate load vectors produced by water-filling on $((N_t', q_t'))_{t\in[m]}$ on this request sequence.
    Define $\mu_i$ to be the latest online neighbor of offline node $i\in[n]$ on sequence $E'$.
    Without loss of generality, we assume that offline nodes are labeled in order of increasing load induced by policy $\calA$: $\bell_m^{\calA}(1) \le\dots\le\bell_m^{\calA}(n)$.
    By \rkedit{the} construction \rkedit{of} \Cref{app:alg:wf_vs_policy_adapt}, this ensures that $\mu_1\le\dots\le\mu_n$ since offline nodes with lower load are removed from $S_t$ first.
    Offline nodes $i$ with lower $\mu_i$ will have lower loads for the \waterfill allocation, so we also have $\bell_m^{\wf}(1)\le\dots\le\bell_m^{\wf}(n)$.
    For convenience, we define $\bell_m^{\wf}(0) \define \bell_m^{\calA}(0) \define \vec{0}$.
    We prove majorization with induction on increasing cumulative sums: $\bell_m^{\wf}([i]) \ge \bell_{m}^{\calA}([i])$ for all $i\in[n]\cup\{0\}$.
    The base case $\bell_m^{\wf}(0) \ge \bell_{m}^{\calA}(0)$ holds because 
    both sides equal zero by definition. For the induction step,
    we assume the inequality holds for $j \le i - 1$, where $i\in[m]$,
    and we aim to prove $\bell_m^{\wf}([j]) \ge \bell_{m}^{\calA}([j])$.
    Since the last neighbor of $i$ is online node $\mu_i$, it suffices to focus on the intermediate load vector after the arrival of that online node.
    The total load across all offline nodes under \waterfill is $\sum_{t=1}^{\mu_j} q_t$ with $\bell_{m}^{\wf}([i-1])$ load allocated to nodes in $[i-1]$.
    Since $((N_t', q_t'))_{t\in[m]}$ is a nested sequence, \Cref{def:height} and \Cref{obs:wf_on_nest} imply that \waterfill has equal load across the remaining nodes $k\ge i$:
    \begin{align}\label{app:eq:deviat1}
        \bell_{\mu_i}^{\wf}(i) = \left(\frac{1}{n-i+1}\right)\left(\sum_{t=1}^{\mu_i} q_t - \bell_m^{\wf}([i-1]) \right)
    \end{align}
    Now we analyze the allocation made by policy $\calA$.
    Again, the total load on all offline nodes is $\sum_{t=1}^{\mu_j} q_t$ and the load allocated to nodes in $[i-1]$ is $\bell_{m}^{\calA}([i-1])$.
    Since offline node $i$ has the least load of nodes in $[n]\setminus[i-1]$ after the arrival of online node $\mu_i$ by construction, we know that:
    \begin{align}\label{app:eq:deviat2}
        \bell_{\mu_i}^{\calA}(i) \le \left(\frac{1}{n-i+1}\right)\left(\sum_{t=1}^{\mu_i}q_t - \bell_{m}^{\calA}([i-1])\right)
    \end{align}
    Comparing cumulative sums proves the hypothesis holds for $i=j$ and implies the claim:
    \begin{align*}
        \bell_m^{\wf}([i])
        &= \bell_{\mu_i}^{\wf}([i]) \\
        &= \bell_{\mu_i}^{\wf}([i-1]) + \left(\frac{1}{n-i+1}\right)\left(\sum_{t=1}^{\mu_i} q_t - \bell_m^{\wf}([i-1]) \right)         
        &&\text{\Cref{app:eq:deviat1}} \\
        &= \left(\frac{n-i}{n-i+1}\right) \bell_{m}^{\wf}([i-1]) + 
           \left(\frac{1}{n-i+1}\right) \sum_{t=1}^{\mu_i} q_t
        &&\text{since $\bell_{\mu_i}^{\wf}([i-1])=\bell_m^{\wf}([i-1])$}\\
        &\ge \left(\frac{n-i}{n-i+1}\right)\bell_{m}^{\calA}([i-1]) + \left(\frac{1}{n-i+1}\right)
        \sum_{t=1}^{\mu_i} q_t
        &&\text{Inductive assumption} \\
        &= \bell_{\mu_i}^{\calA}([i-1]) + \left(\frac{1}{n-i+1}\right)
        \left( \sum_{t=1}^{\mu_i} q_t - \bell_{m}^{\calA}([i-1]) \right)
        &&\text{since $\bell_{\mu_i}^{\calA}([i-1])=\bell_m^{\calA}([i-1])$} \\
        &\ge \bell_{\mu_i}^{\calA}([i])
        &&\text{\Cref{app:eq:deviat2}} \\
        &= \bell_{m}^{\calA}([i])
    \end{align*}
    The inductive hypothesis holds for an arbitrary realization of the request sequence and loads produced by $\calA$, thus the original claim holds.
\end{proof}

\subsection{Majorization Minimality and Minimax Optimality of \waterfill}
\wfminadap*
\begin{proof}
    We compose \Cref{lem:nest_wf,app:lem:wf_vs_polict_adapt} to prove the claim.
    Given a request sequence $E\in\calE_{n,m,q}$, apply \Cref{lem:nest_wf} to get $\wt{E}\in\Enest_{n,m,q}$ then apply \Cref{app:lem:wf_vs_polict_adapt} to $\wt{E}$ to get adaptive sequence $E'\in\calE_{n,m,q}$.
    Notably, since \waterfill is a deterministic algorithm, \Cref{app:fact:adapt_det} says the optimal hindsight solution does not depend on $\wf$: $\opt(E) = \opt(\wf,E)$.
    These sequences guarantee both of the following with probability $1$:
    \begin{align*}
        &\wf(E) \preceq \wf(\wt{E}) \preceq \wf(E') \\
        &\opt(\E) \succeq \opt(\wt{E}) \sim \opt(\calA, E')
    \end{align*} 
\end{proof}

\wfadapt*
\begin{proof}
    We prove the claim for Schur-concave $f$ maximization.
    The proof for Schur-convex $g$ minimization is identical up to 
    sign reversals necessitated by the reversal of the directionality of optimization.
    For any policy $\calA$ and request sequence $E$, it suffices to show that there is another adaptive request sequence $E'$ on which the cost of $\calA$ on $E'$ is no less than that of $\wf$ on $E$. 
    The claim is a result of \Cref{app:thm:wf_min_adapt} and the definition of Schur-concavity:
    \begin{align*}
        c_{\alpha, f}^{\max}(\wf, E)
        &= \E[\alpha \cdot f(\opt(\wf,E)) - f(\wf(E))] \\
        &\le \E[\alpha \cdot f(\opt(\calA, E')) - f(\calA(E'))]
        &&\text{\Cref{app:thm:wf_min_adapt}, Schur-concavity} \\
        &= c_{\alpha, f}^{\max}(\calA, E')
    \end{align*}
\end{proof}

\section{Deferred Proofs}
\label{appsec:proofs}

\subsection{Additional Details for~\Cref{sec:def}}
\label{appsec:def}

For completeness, we provide a brief proof outline for~\Cref{fact:opt}, based on the theory of submodular optimization on polymatroids (see \cite[Chapter II.]{fujishige2005submodular} for details).

\begin{proof}[Proof Outline for~\Cref{fact:opt}]
    For each $t\in[m]$ define the submodular function $f_t(A) \define q_t \cdot \1{A\cap N_t \neq \emptyset}$.
    By definition, the set $\Delta(N_t, q_t)$ is the base polytope of the polymatroid with rank function $f_t$.
    \cite[Equation 3.33]{fujishige2005submodular} shows that the Minkowski sum of these base polytopes, which is equal to the set of feasible allocations $\Delta(E)$, is the base polytope of the polymatroid with rank function $f = \sum_{t\in[m]} f_t$.
    Finally, \cite[Theorem 3]{dutta1989concept} proves the existence of a majorization minimal convex game core, which is equivalent to a majorization minimal polymatroid base (\cite[Subsection 2.2 Polymatroids]{fujishige2005submodular}).
\end{proof}

\subsection{Deferred Proofs from~\Cref{sec:result}}
\label{appsec:result}

Next we provide the deferred proofs of the main results and corollaries we state in~\Cref{sec:result}.

\subsubsection*{Minimax Optimality of \waterfill}

First, we formally establish our main result,~\Cref{thm:wf_min}, which, as we mentioned before, follows immediately from \Cref{lem:nest_wf,lem:wf_vs_policy}, which we establish in~\Cref{sec:nest,sec:wf_vs_policy} respectively.
\wfmin*
\begin{proof}
    Let $\widetilde{E}\in\calE_{n,m,q}$ be the sequence output by \Cref{alg:nest} on $E$, and $E'\in\calE_{n,m,q}$ be the sequence output by \Cref{alg:wf_vs_policy} on $\widetilde{E}$.
    The transitivity of majorization implies the claim:
    \begin{align*}
        &\wf(E) \preceq \wf(\widetilde{E}) \preceq \E[\calA(E')]
        &&\text{Lemmas \ref{lem:nest_wf} and \ref{lem:wf_vs_policy}}\\
        &\opt(E) \succeq \opt(\widetilde{E}) \sim \opt(E')
        &&\text{Lemmas \ref{lem:nest_wf} and \ref{lem:wf_vs_policy}}
    \end{align*}
\end{proof}

Next we use this to get our stated implications for the regret of \waterfill under particular objectives, and compared to both deterministic and randomized policies. The first result follows directly from the definition of Schur-monotone functions.

\determinimax*
\begin{proof}
    It suffices to show that on any sequence $E\in\calE_{n,m,q}$, the cost incurred by the principal when she plays $\wf$ is less than that of $\calA$ on some sequence $E'\in\calE_{n,m,q}$. 
    Fix such a sequence $E\in\calE_{n,m,q}$.
    \Cref{thm:wf_min} guarantees that there is $E'\in\Enest_{n,m,q}$ such that $\wf(E) \preceq \calA(E')$ and $\opt(E) \succeq \opt(E')$.
    Note that the expectation is omitted because $\calA$ is deterministic.
    Schur-concavity of $f$ guarantees $f(\wf(E)) \ge f(\calA(E'))$ and $f(\opt(E)) \le f(\opt(E'))$.
    This yields the desired result:
    \begin{align*}
        c_{\alpha,f}^{\max}(\wf, E) = \alpha f(\opt(E)) - f(\wf(E)) \le \alpha \cdot f(\opt(E')) - f(\calA(E')) = c_{\alpha, f}^{\max}(\calA, E')
    \end{align*}
    Applying the \Cref{thm:wf_min} similarly for Schur-convex objectives $g$ completes the claim. 
\end{proof}

For comparing against randomized algorithms, we need to add the requirement of convexity/concavity, and then leverage Jensen's inequality, as follows.

\randminimax*

\begin{proof}
    Similarly to \Cref{cor:reg_rand}, it suffices to show that for every policy $\calA$ and sequence $E\in\calE_{n,m,q}$, there is a sequence $E'\in\calE_{n,m,q}$ such that the cost of \waterfill on $E$ is less than the cost of $\calA$ on $E'$.
    Fix a sequence $E\in\calE_{n,m,q}$.
    \Cref{thm:wf_min} guarantees that there is $E'\in\Enest_{n,m,q}$ such that $\wf(E) \preceq \E[\calA(E')]$ and $\opt(E) \succeq \opt(E')$.
    Consider when $f$ is symmetric and concave.
    $f$ is also Schur-concave by \cite[Proposition C.2]{marshall2011majorization}.
    Applying the definition of Schur-concavity and Jensen's inequality yields:
    \begin{align*}
        &f(\wf(E)) \ge f(\E[\calA(E')]) \ge \E[f(\calA(E'))]
    \end{align*}
    Further, we have $f(\opt(E)) \le f(\opt(E'))$.
    These two facts imply the desired inequality:
    \begin{align*}
        c_{\alpha,f}^{\max}(\wf, E) = \alpha f(\opt(E)) - f(\wf(E)) \le \alpha \cdot f(\opt(E')) - \E[f(\calA(E'))] = c_{\alpha, f}^{\max}(\calA, E')
    \end{align*}
    Applying similar arguments for symmetric and convex $g$ completes the corollary.
\end{proof}

Finally, we establish the separation the performance of \waterfill and randomized policies for general Schur-monotone functions under oblivious adversaries. 

\needconcave*

\begin{proof}
    We first show $f$ is Schur-concave.
    Consider $\mbf{x}\preceq \mbf{y}$ and assume $\mbf{x}(1) \le \mbf{x}(2)$ by symmetry.
    The increasing sum definition of majorization implies that $\mbf{x}(1) \ge \mbf{y}(1)$ and $\mbf{x}(2) \le \mbf{y}(2)$.
    If $\mbf{x}(1) > 1/2$, then $f(\mbf{x}) = 1 \ge f(\mbf{y})$ as desired.
    Otherwise, $1/2 \ge \mbf{x}(1) \ge \mbf{x}(2)$ and $f(\mbf{x}) \ge 0 = f(\mbf{y})$.
    This shows that $f$ is Schur-concave.
    
    We show that the $\alpha$-regret of \waterfill under objective $f$ is $\alpha$.
    Take the sequence $E = ((N_1, q_1), (N_2, q_2))\in\calE_{2,2,2}$ with $N_1 = \{1,2\}$, $N_2 = \{2\}$, and $q_1 = q_2 = 1$.
    The \waterfill allocation is $\wf(E) = (1/2, 3/2)$ and optimal hindsight allocation is $\opt(E) = (1,1)$.
    The $\alpha$-regret of $\wf$ lower bounded by:
    \begin{align*}
        \reg_{(2,2,2),\alpha,f}^{\max}(\wf)
        = \alpha f(\opt(E)) - f(\wf(E))
        = \alpha
    \end{align*}

    We now describe the randomized allocation policy $\calA$, which provides a better expected competitive ratio.
    Initially, the policy uniformly at random selects a primary offline node $u\gets \uniform([2])$ and defines the secondary offline node $v\gets [2]\setminus \{i\}$.
    Consider the $t^{th}$ online node with neighborhood $N_t$ and quantity $q_t$.
    There are three possible neighborhoods: full neighborhoods $N_t = \{1, 2\}$ and the singleton neighborhoods $N_t\in\{\{1\}, \{2\}\}$.
    If $N_t$ is a singleton neighborhood, the entire quantity of the online node is allocated to its single neighbor: $N_t = \{i\}$ implies $\mbf{x}(i) = q_t$ and $\mbf{x}([n]\setminus \{i\}) = 0$.
    If a full neighborhood $N_t = \{1, 2\}$ arrives, the policy allocates quantity to the primary offline node until it reaches $3/4$ load.
    It allocates any remaining quantity to the secondary node.
    Formally, the allocation of the $t$ offline node is $\mbf{x}_t(u) \gets \min(q_t, \pp{3/4 - \bell_{t-1}(u)})$ and $\mbf{x}_t(v) \gets q_t - \mbf{x}_t(v)$, where $\bell_{t-1}$ is the load vector prior to the arrival of $t$.

    We complete the proof by showing that the cost incurred by the policy $\calA$ is at most $\pp{\alpha - 1/2}$.
    Since $\opt(E)\preceq \wf(E)$ by majorization minimality of $\opt$, when the adversary plays $E\in\calE_{2,2,2}$ with $f(\opt(E)) = 0$ the equity of the \waterfill allocation is $f(\wf(E)) = 0$;
    the cost $c_{2,2,2}^{\max}(\calA, E)$ incurred by the principal for such sequences is $0$.
    
    We now analyze the other case when the adversary plays $E\in\calE_{n,m,q}$ with $f(\opt(E)) = 1$.
    Define $\mathscr{q}_N \define \sum_{t\in[m]: N_t = N} q_t$ to be the total quantity of online nodes with neighborhood $N$.
    We break up our analysis into case work.
    We focus on when there is an offline node $v$ with $\mathscr{q}_{\{v\}} > 1/2$ and the policy selects primary node $u = [2]\setminus \{v\}$ with probability $1/2$.
    By definition $\calA(E)(v) \ge \mathscr{q}_{\{v\}} > 1/2$.
    On the other hand, tie-breaking in favor of the primary node yields $\calA(E)(u) \ge \min(\mathscr{q}_{\{u\}} + \mathscr{q}_{\{u,v\}}, 3/4)$.
    The total quantity incident on $u$ must be at least $\mathscr{q}_{\{u\}} + \mathscr{q}_{\{u,v\}}> 1/2$ because $f(\opt(E)) = 1$ so we have $\calA(E)(u) > 1/2$.
    In conclusion, $\E[f(\calA(E))] \ge 1/2$ because $f(\calA(E)) = 1$ with probability $1/2$ and $f(\calA(E)) \ge 0$ otherwise.
    The cost incurred by the principal in this case is $\alpha - 1/2$.

    Consider the other case when $\mathscr{q}_{\{1\}} \le 1/2$ and $\mathscr{q}_{\{2\}} \le 1/2$.
    The total quantity on online nodes is $2$ so $\mathscr{q}_{\{1,2\}} = 2 - \mathscr{q}_{\{1\}} - \mathscr{q}_{\{2\}} \ge 1$.
    By symmetry, assume the policy selects primary node $u\gets 1$.
    Favorable tie-breaking means the primary node satisfies:
    \begin{align*}
        \calA(E)(1)
        &\ge \min(\mathscr{q}_{\{1\}} + \mathscr{q}_{\{1,2\}}, 3/4) \\
        &\ge \min(1, 3/4)
        && \mathscr{q}_{\{1,2\}} \ge 1 \\
        &= 3/4
    \end{align*}
    The primary node $1$ receives at most $3/4$ load from nodes $t$ with $N_t = \{1,2\}$.
    The load on the secondary node is:
    \begin{align*}
        \calA(E)(2)
        &\ge \mathscr{q}_{\{2\}} + (\mathscr{q}_{\{1,2\}} - 3/4) \\
        &= 5/4 - \mathscr{q}_{\{1\}}
        &&\mathscr{q}_{\{1\}} + \mathscr{q}_{\{2\}} + \mathscr{q}_{\{1,2\}} = 2 \\
        &\ge 3/4
        &&\mathscr{q}_{\{1\}} \le 1/2
    \end{align*}
    For these sequences, the policy guarantees $f(\calA(E))$  almost surely and the cost incurred by the principal is $\alpha - 1$

    Since the adversary plays the worst-case sequence, regret is bounded by the maximum cost of the three analyzed cases: 
    \begin{align*}
        \reg_{(2,2,2),\alpha,f}^{\max}(\calA) = \max(0, \alpha - 1/2, \alpha - 1) = \pp{\alpha - 1/2}
    \end{align*}

\end{proof}

\subsubsection*{Characterizing the Minimax Regret}

For convenience, we henceforth define
$
\calI_{n,q}=\{\mbf{x}\in\R_{>0}^n : \mbf{x}([n])=q \text{ and } \mbf{x}_1\le\cdots\le \mbf{x}_n\}
$ to be the subset of the $q$-simplex with coordinatewise non-decreasing vectors.

\regret*
\begin{proof}
    We prove the claim for the case of maximizing Schur-concave $f$.
    Applying similar arguments shows the results hold for minimizing Schur-concave $g$.
    To begin, we provide a lower bound on the $\alpha$-regret of \waterfill.
    From \Cref{obs:worst_exist}, we know that for all $\bell\in\calI_{n,q}$, there is $E\in\calE_{n,m,q}$ with $\wf(E) = H\bell$ and $\opt(E) = \bell$.
    This implies that $\alpha$-regret satisfies:
    \begin{align*}
        \reg_{(n,:,q),\alpha,f}^{\max}(\wf)
        = \sup_{E\in\calE_{n,:,q}} (\alpha\cdot f(\opt(E)) - f(\wf(E)))
        \ge \sup_{\bell \in\calI_{n,q}} (\alpha\cdot f(\bell) - f(H\bell))
    \end{align*}
    The lower bound is a result of composing \Cref{thm:wf_min,thm:param}.
    Fix an arbitrary $E\in\calE_{n,:,q}$.
    \Cref{thm:wf_min} shows that there is $\widetilde{E}\in\Enest_{n,:,q}$ satisfying $\wf(E) \preceq \wf(\widetilde{E})$ and $\opt(E) \succeq \opt(\widetilde{E})$.
    Further, applying \Cref{thm:param} to $\widetilde{E}$ proves the existence of $\bell\in\calI_{n,q}$ with $\wf(\widetilde{E}) \preceq H\bell$ and $\opt(\widetilde{E}) \succeq \bell$.
    Composing these theorems and using the definition of Schur-concavity gives:
    \begin{align}\label{eq:bell_exist}
        f(\wf(E)) \ge f(\wf(\widetilde{E})) \ge f(H\bell)
        \quad\quad\text{and}\quad\quad
        f(\opt(E)) \le f(\opt(\widetilde{E})) \le f(\bell)
    \end{align}
    The existence of such an $\bell$ provides an upper bound on $\alpha$-regret:
    \begin{align*}
        \reg_{(n,:,q),\alpha,f}^{\max} (\wf)
        &= \sup_{E\in\calE_{n,:,q}} (\alpha\cdot f(\opt(E)) - f(\wf(E))) \\
        &\le \sup_{E\in\calE_{n,:,q}} \sup_{\bell\in\calI_{n,q}} (\alpha\cdot f(\opt(E)) - f(\wf(E)))
        &&\text{\Cref{eq:bell_exist}} \\
        &= \sup_{\bell\in\calI_{n,q}} (\alpha\cdot f(\opt(E)) - f(\wf(E)))
    \end{align*}
\end{proof}

\subsection{Deferred Proofs and Additional Details  from~\Cref{sec:nest}}
\label{appsec:lemmas}

\subsubsection{\waterfill Invariance under Active Edge Pruning}

\inactivecomp*

\begin{proof}
    Let $\bell_1,\dots,\bell_m$ be the intermediate loads produced by $\wf$ on $E$.
    We now show that the intermediate loads $\bell_1^{(1)},\dots,\bell_m^{(1)}$ induced by $\wf$ on sequence $E^{(1)}$ is equivalent to $\bell_1,\dots,\bell_m$.
    For induction, assume that $\bell_{t-1} = \bell_{t-1}^{(1)}$.
    We define $\bell_0 = \bell_0^{(1)} = \vec{0}$ to support the base case of the inductive hypothesis for $t=1$.
    Consider the $\wf$ optimization problem (\Cref{def:wf}) on round $t$.
    By definition of $E^{(1)}$ we have $\Delta(N_t, q_t) \supseteq \Delta(N_t^{(1)}, q_t^{(1)})$ and $\mbf{x}_t\define \wf(N_t, q_t, \bell_{t-1}) \in \Delta(N_t^{(1)}, q_t^{(1)})$.
    The inductive hypothesis shows that $\mbf{x}_t + \bell_{t-1} = \mbf{x}_t + \bell_{t-1}^{(1)}$ is the majorization minimal intermediate allocation on round $t$, so \Cref{obs:wf_properties} implies $\mbf{x}_t = \wf(N_t^{(1)}, q_t^{(1)}, \bell_{t-1}^{(1)})$.
    The inductive hypothesis holds and implies the claim: $\wf(E) = \bell_m = \bell_m^{(1)} = \wf(E^{(1)})$.
\end{proof}

The proof of~\Cref{prop:inactive_comp} also indicates that the heights of online nodes under \waterfill are equivalent for the two sequences. This allows us to exploit the following fact about neighborhoods of offline nodes in $E^(1)$, which is critical for the invariance of the \waterfill allocation following the online node permutation in Step 2 (i.e., for~\Cref{prop:neighborhood_pres}).
\begin{observation}[Monotone Neighborhoods]
\label{obs:mono_neighborhood}
Let $E\in\calE_{n,m,q}$ be an sequence without inactive edges, and let $(h_t)_{t\in[m]}$ be the heights under \waterfill of online nodes in $E$. 
Then neighbors of each offline node $i$ appear in order of increasing height, i.e., $\forall\,t,s\in\Gamma_i(E^{(1)})$ with $t < s$, we have $h_t < h_s$.
\end{observation}

\begin{proof}
    Let $(\mbf{x}_{t})_{t\in[m]}$ and $(\bell_t)_{t\in[m]}$ be allocations and intermediate loads produced by the \waterfill algorithm, respectively.
    Consider an offline node $i$ which has neighbors $t,s\in \Gamma_i(E)$ with $t < s$.
    The height of online node $t$ is $h_t = \bell_t(i)$ and the height of online node $s$ is $h_s = \bell_s(i)$.
    By definition, we have:
    \begin{align*}
        h_s = \bell_s(i) = \bell_t(i) + \sum_{u = t+1}^{s} \mbf{x}_u(i) = h_t + \sum_{u=t+1}^{s} \mbf{x}_u(i)
    \end{align*}
    The above sum is strictly positive since $i$ has an active edge to $s$; this completes the proof.
\end{proof}

\subsubsection{\waterfill Invariance under Online Node Permutation}

\begin{proposition}[\waterfill preserving permutations]
\label{prop:reorder}
    Fix an sequence $E = ((N_t, q_t))_{t\in[m]}\in\calE_{n,m,q}$.
    Take $\sigma:[m]\to[m]$ to be any permutation satisfying $t<s\Rightarrow \sigma(t) < \sigma(s)$ for every offline node $i\in [n]$ and every pair $t,s\in \Gamma_i(E)$.
    Define the sequence $\widehat{E}\gets((\widehat{N}_t, \widehat{q}_t))_{t\in[m]}\in\calE_{n,m,q}$ with $\widehat{N}_{t} \gets N_{\sigma^{-1}(t)}$ and $\widehat{q}_{t} \gets q_{\sigma^{-1}(t)}$.
    Then the allocations $(\mbf{x}_t)_{t\in[m]}$ and $(\widehat{\mbf{x}}_t)_{t\in[m]}$ produced by \waterfill on $E$ and $\widehat{E}$, respectively, satisfy $\mbf{x}_{t} = \widehat{\mbf{x}}_{\sigma(t)}$ for all $t\in[m]$, and hence $\wf(E) \sim \wf(\widehat{E})$.
\end{proposition}

\begin{proof}
    We prove this claim via induction on online nodes.
    Let $(\mbf{x}_t)_{t\in[m]}$ and $(\bell_t)_{t\in[m]}$ be the allocations and intermediate allocations generated by \waterfill on $E$, respectively.
    Define $(\widehat{\mbf{x}}_t)_{t\in[m]}$ and $(\widehat{\bell}_t)_{t\in[m]}$ similarly for $\widehat{E}$.
    For induction, assume that $\mbf{x}_t = \widehat{\mbf{x}}_{\sigma(t)}$ and $\bell_t =_{N_t} \widehat{\bell}_{\sigma(t)}$ for all $t\le s-1$.
    To simplify analysis, it is useful to define the initial allocation $\mbf{x}_0 \define \widehat{\mbf{x}}_0 \define \vec{0}$ and intermediate allocation $\bell_0 \define \widehat{\bell}_0 \define \vec{0}$ as well as extend the bijection $\sigma$ so that $\sigma(0) = 0$.
    These definitions imply the inductive assumption for $s=1$.

    Consider an arbitrary round $s$.
    For each $i\in N_s$ define $\tau_i \define\max\left((\Gamma_i \cap [s-1])\cup \{0\}\right)$ to be the last online node prior to $s$ which was a neighbor of $i$ (if no such neighbor exists, we set $\tau_i \define 0$).
    Since $\sigma$ preserves the ordering of the neighborhood of each offline node, we know that $\sigma(\tau_i)$ is the latest arriving online node before $\sigma(s)$ with $i$ as a neighbor.
    Thus, $\widehat{\bell}_{\sigma(s)-1}(i) = \widehat{\bell}_{\sigma(\tau_i)}(i)$ for each $i\in \widehat{N}_t$ and we have the following equivalence for all $i\in N_t = \widehat{N}_{\sigma(t)}$:
    \begin{align}\label{eq:worst_case_ind1}
        \bell_{s-1}(i)
        &= \bell_{\tau_i}(i)
        &&\text{Definition of }\tau_i \\
        &= \widehat{\bell}_{\sigma(\tau_i)}(i)
        &&\text{Inductive assumption on }\tau_i < s \nonumber\\
        &= \widehat{\bell}_{\sigma(s)-1}(i)
        &&\sigma\text{ preserves neighborhoods of offline nodes} \nonumber
    \end{align}
    Consider the fractional allocation of \waterfill for online node $s$ on $E$ and online node $\sigma(s)$ on $\widehat{E}$.
    The objectives given in \Cref{def:wf} satisfy:
    \begin{align*}
        &\text{Obj. defining }\wf(N_s, q_s, \bell_{s-1}) \\
        &= \min_{i\in N_t}\left\{\mbf{x}_s(i) + \bell_{s-1}(i)\right\} \\
        &= \min_{i\in \widehat{N}_{\sigma(t)}}\left\{ \mbf{x}_s(i) + \widehat{\bell}_{\sigma(s) - 1}(i) \right\}
        &&N_t = \widehat{N}_{\sigma(t)}\text{ and \Cref{eq:worst_case_ind1}} \\
        &= \text{Obj. defining }\wf(\widehat{N}_{\sigma(s)}, \widehat{q}_{\sigma(s)}, \widehat{\bell}_{\sigma(s)-1})
    \end{align*}
    The feasible regions are also equivalent because $(N_s, q_s) = (\widehat{N}_{\sigma(s)}, \widehat{q}_{\sigma(s)})$.
    Since the optimization problem specifying the allocations on $E$ and $\widehat{E}$ are equivalent and this optimization problem has a unique solution (\Cref{obs:wf_properties}), we have $\mbf{x}_s = \wf(N_s, q_s, \bell_{s-1}) = \wf(\widehat{N}_{\sigma(s)}, \widehat{q}_{\sigma(s)}, \widehat{\bell}_{\sigma(s)-1}) = \widehat{\mbf{x}}_{\sigma(s)}$.
    The second part of the inductive assumption also holds:
    \begin{align*}
        \bell_{s}
        &= \mbf{x}_s + \bell_{s-1} \\
        &=_{N_t} \widehat{\mbf{x}}_{\sigma(s)} + \widehat{\bell}_{\sigma(s)-1}
        &&\mbf{x}_{s} = \widehat{\mbf{x}}_{\sigma(s)} \text{ and \Cref{eq:worst_case_ind1}}\\
        &= \widehat{\bell}_{\sigma(s)}
    \end{align*}

    Assume the inductive claim on all $t\in [m]$.
    To prove that $\bell_m = \widehat{\bell}_m$, define $\mu_i \gets \max(\Gamma_i(E))$ to be the last online node with an edge to $i$.
    We use the property that $\sigma$ is order-preserving on elements in $\Gamma_i(E)$ to show the desired result.
    For all offline nodes $i$, we have:
    \begin{align*}
        \bell_m(i) &= \bell_{\mu_i}(i)
        &&\sigma\text{ preserves order} \\
        &= \widehat{\bell}_{\sigma(\mu_i)}(i)
        &&\text{Inductive assumption} \\
        &= \widehat{\bell}_{m}(i)
        &&\sigma\text{ preserves order}
    \end{align*}
    We have shown $\mbf{x}_t = \widehat{\mbf{x}}_{\sigma(t)}$ for all $t$ and $\wf(E) = \wf(\widehat{E})$.
\end{proof}

\nbdpreserve*
\begin{proof}
    To prove this result, it suffices to show that $\sigma$ preserves the relative order of offline node neighborhoods: for each offline node $i$ we are guaranteed that $\sigma(t) < \sigma(s)$ for $t < s$ with $t,s\in \Gamma_i(E^{(1)})$.
    This then allows us to apply \Cref{prop:reorder} and get $\wf(E^{(1)}) = \wf(E^{(2)})$.

    Let $(h_t)_{t\in[m]}$ be the \waterfill heights of online nodes in sequence $E$, which is input to \Cref{alg:nest}.
    By \Cref{prop:inactive_comp}, the \waterfill heights of $E^{(1)}$ are also $(h_t)_{t\in[m]}$.
    Fix an offline node $i$ and consider two of its neighbors $t,s\in\Gamma_i(E^{(1)})$ with $t<s$.
    \Cref{obs:mono_neighborhood} implies $h_t < h_s$ so \Cref{alg:nest} will construct the bijection $\sigma$ so that $\sigma(t) < \sigma(s)$.
    Thus, $\sigma$ preserves the relative order of each offline node's neighborhood as desired.
\end{proof}
\subsubsection{\Cref{alg:nest} Preserves Hindsight Allocations}

Finally, we provide a proof for~\Cref{prop:optpreserve}. For this, we first need the following observation about inactive edges.

\begin{observation}[Inactive Edge Tracking]\label{obs:hind_track}
    Fix a sequence $E\in\calE_{n,m,q}$ which has \waterfill heights $(h_t)_{t\in[m]}$.
    Let $(i,t)$ be an edge that is inactive for \waterfill.
    Then there is an edge $(i, s)$ which is active for \waterfill and satisfies at least one of 1) $h_t < h_s$ or 2) $h_t = h_s$ and $s < t$.
\end{observation}

\begin{proof}
    Let $(\mbf{x}_t)_{t\in[m]}$ and $(\bell_t)_{t\in[m]}$ be the allocations and intermediate load vectors of \waterfill on $E$, respectively.
    Fix an inactive edge $(i,t)$ and define $s \define \max(\{s\in \Gamma_t(E)\mid \mbf{x}_s(i) > 0\}\cap [t-1])$ to be the last arriving online neighbor with an active edge to $i$ before online node $t$.
    Since $s$ is the last active edge to $i$, its height is $h_s = \bell_{t-1}(i)$.
    We finish the proof via case work.
    Edge $(i,t)$ is inactive, so $\bell_{t-1}(i) \ge h_t$ otherwise \waterfill would have allocated some positive value to edge $(i,t)$.
    In the case that $\bell_{t-1}(i) > h_t$, our previous arguments show $h_s = \bell_{t-1}(i) > h_t$ as desired.
    Consider the second case when $\bell_{t-1}(i) = h_t$.
    Then $h_s = \bell_{t-1}(i) = h_t$ and $s < t$ by definition.
    This completes the proof.
\end{proof}

\optpreserve* 

\begin{proof}
    Consider an edge $(i,t)$ in the sequence $E$ that is active for \waterfill and the hindsight solution.
    The edge $(i,\sigma(t))$ will be in sequence $\widetilde{E}$ by construction, so the positive allocation $\widetilde{\mbf{x}}_{\sigma(t)}(i) = \mbf{x}_t > 0$ does not violate any constraints.

    Finally, we examine edges $(i,t)$ in sequence $E$ which are inactive for \waterfill but active for the hindsight solution.
    By \Cref{obs:hind_track}, there must be an edge $(i,s)$ in $E$ which is active for \waterfill and satisfies one of the following 1) $h_t < h_s$ or 2) $h_t = h_s$ and $s<t$.
    In either case, $\sigma(t) < \sigma(s)$ by construction and $(i, \sigma(i))$ will be in sequence $E^{(2)}$.
    When performing the nestification operation on $E^{(2)}$, the latest online neighbor of offline node $i$ will be at least $\mu_i = \max(\Gamma_i(E^{(2)}) \cup\{0\}) \ge \sigma(s)$.
    The nested sequence $E^{(3)}$ will have an edge from online node $\sigma(t)$ to offline node $i$ since $\sigma(t) < \sigma(s) \le \mu_i$. 
    The existence of edge $(i, \sigma(t))$ in $\widetilde{E} = E^{(3)}$ means the positive allocation $\widetilde{\mbf{x}}_{\sigma(t)}(i) = \mbf{x}_t(i) > 0$ does not violate any constraints.
    
    The allocation $(\mbf{x}_t)_{t\in[m]}$ only has positive allocations on edges in $\widetilde{E}$, so it is compatible.
    Further, the allocation is full distribution $\widetilde{\mbf{x}}_t([i]) = \mbf{x}_{\sigma^{-1}(t)}([n]) = q_{\sigma^{-1}(t)} = \widetilde{q}_t$, where $\widetilde{q}_t$ is the quantity of online node $t$ in sequence $\widetilde{E}$.
    The allocation $(\widetilde{\mbf{x}}_t)_{t\in[m]}$ is feasible on $\widetilde{E}$, so the proof is complete.
\end{proof}

\subsection{Deferred Proofs from~\Cref{sec:uppertriangle}}
\label{appsec:uppertriangle}

We first generalize the relation $\wf(E)=Hq$ for the \waterfill allocation on a complete upper triangular sequence $E$, to the more general statement about nested sequences

\hwf*

\begin{proof}
    Let $(\mbf{x}_t)_{t\in[m]}$ and $(\bell_t)_{t\in[m]}$ be the \waterfill allocation and and intermediate loads, respectively.
    Since $E$ is nested, the intermediate load $\bell_{t-1}(i)$ of offline nodes $i\in N_t$ are all equivalent.
    Thus, the \waterfill allocation admits a closed-form solution:
    \begin{align}\label{eq:wf_sol}
        \mbf{x}_t(i) = \frac{q_t\cdot\1{i\in N_t}}{|N_t|}
    \end{align}
    We finish the claim via direct calculation.
    In our arguments, we use the fact that $\{\Gamma_i(E)\setminus \Gamma_i(E)\}_{i\in[n]}$ are disjoint sets, as a result of the nestedness of $E$.
    \begin{align*}
        \sum_{t\in[m]} \mbf{x}_t(j)
        &= \sum_{i\in[n]} \left(\sum_{t\in \Gamma_i(E)\setminus \Gamma_{i-1}(E)} \mbf{x}_t(j)\right)
        &&\text{Sets }\{\Gamma_i(E)\setminus \Gamma_i(E)\}_{i\in[n]}\text{ are disjoint} \\
        &= \sum_{i\in[n]} \left(\sum_{t\in \Gamma_i(E)\setminus \Gamma_{i-1}(E)} \frac{q_t\cdot\1{j\in N_t}}{|N_t|}\right)
        &&\text{\Cref{eq:wf_sol}} \\
        &= \sum_{i\in[j]} \left(\sum_{t\in \Gamma_i(E)\setminus \Gamma_{i-1}(E)} \frac{q_t}{n-i+1}\right)
        &&\text{Nestedness} \\
        &= \sum_{i\in[j]} \frac{\mbf{z}(i)}{n-i+1}
        &&\text{Definition of }\mbf{z}
    \end{align*}
    The above is equivalent to $\wf(E) = \sum_{t} \mbf{x}_t H\mbf{z}$
\end{proof}

Next, we prove the majorization preservation property of the $H$ transform.
For convenience, we henceforth define
$
\calI_{n,q}=\{\mbf{x}\in\R_{>0}^n : \mbf{x}([n])=q \text{ and } \mbf{x}_1\le\cdots\le \mbf{x}_n\}.
$

\hmajpres*

\begin{proof}
    This proof was generated in conversation with Gemini \cite{gemini2025}; see \Cref{app:sec:disclose} for more details.
    Notice that $H\mbf{x}, H\mbf{y}\in \calI_{n,q}$.
    We complete the claim by using the increasing cumulative sum definition of majorization.
    Consider the sum of the smallest $k$ elements in $Hx$.
    \begin{align*}
        \sum_{i=1}^{k} \sum_{j=1}^{n} H(i,j)\mbf{x}(j)
        &= \sum_{i=1}^{k} \sum_{j=1}^{n} \frac{\mbf{x}(j) \cdot \1{i\ge j}}{n-j+1}
        &&\text{Definition of }H \\
        &= \sum_{j=1}^{k} \sum_{i=j}^{k} \frac{\mbf{x}(j)}{n-j+1}
        &&\text{Rearrange sums} \\
        &= \sum_{j=1}^{k} \left(\frac{k-j+1}{n-j+1}\right)\mbf{x}(j)
    \end{align*}
    Define $w_j \define \left(\frac{k-j+1}{n-j+1}\right)$ and $w_{k+1} = 0$; notice that $w_1 \ge\dots\ge w_{k+1}$.
    We now express each scalar $w_j$ as a telescoping series and group like terms:
    \begin{align*}
        \sum_{j=1}^{k} w_j \mbf{x}(j)
        &= \sum_{j=1}^{k} \left(\sum_{i=j}^{k} w_i - w_{i+1}\right)\mbf{x}(j)
        &&\text{Reverse telescope} \\
        &= \sum_{i=1}^{k} (w_i - w_{i+1}) \mbf{x}([i])
        &&\text{Rearrange sums}
    \end{align*}
    Similarly, the sum of the $k$ smallest elements in $H\mbf{y}$ is $\sum_{i=1}^{k} (w_i - w_{i+1})\mbf{y}([i])$.
    Since $w_i - w_{i+1} \ge 0$ and $\mbf{x}([i]) \le \mbf{y}([i])$, the increasing cumulative sums of $H\mbf{x}$ are less than that of $H\mbf{y}$.
    Equivalently, $H\mbf{x} \succeq H\mbf{y}$ 
\end{proof}

Finally, for the sake of completeness, we establish that worst-case sequences in~\Cref{thm:param} exist.

\begin{observation}[Existence of Worst-Case Allocation Sequences]\label{obs:worst_exist}
    For arbitrary $\bell\in\calI_{n,q}$, the nested sequence $E = ((N_t, q_t))_{t\in[m]}\in\Enest_{n,m,q}$ with $m=n$, $N_t\gets [n]\setminus[t-1]$, and $q_t\gets \bell(t)$ satisfies $\wf(E) = H\bell$ and $\opt(E) = \bell$. 
\end{observation}

\begin{proof}
    The allocation $\mbf{x}_t^*(i) = q_t\cdot \1{i = t}$ is feasible and satisfies $\bell = \sum_{t\in[m]} \mbf{x}_t^*$.
    We show that this allocation is majorization minimal using the increasing cumulative sum characterization.
    The load on nodes $[j]$ is $\bell([j]) = \sum_{t \le j} q_t$.
    Notice that this sum is also the total quantity of online nodes with a neighbor in $[j]$.
    Thus, no allocation can yield a greater load on $[j]$: $\bell$ is majorization minimal.

    We can find a closed-form solution for $\wf(E)$ using \Cref{obs:h_wf} and the fact that $\Gamma_i(E) = [i]$ by construction.
    For $\mbf{z}$ defined as in \Cref{obs:h_wf}, we have:
    \begin{align*}
        \mbf{z}(i)
        = \sum_{t\in \Gamma_i(E) \setminus \Gamma_{i-1}(E)} q_t
        = q_i = \bell(i)
    \end{align*}
    Applying \Cref{obs:h_wf} gives $\wf(E) = H\mbf{z} = H\cdot \opt(E)$.
\end{proof}

\section{Competitive Ratio Characterizations}
\label{appsec:CR}

We analyze competitive ratios for $n$ offline nodes.
The adversary is allowed to choose any number of online nodes $m$ and any quantity $q$.
We characterized that the competitive ratio of policy $\calA$ for $f$ maximization (resp. $g$ minimization) is the largest (resp. smallest) $\alpha$ such that $\reg_{(n,:,:),\alpha,f}^{\max}(\calA) \le 0$ (resp. $\reg_{(n,:,:),\alpha,g}^{\min}(\calA) \le 0$).
By simple algebraic manipulation, it can be seen that the competitive ratio is equivalent to $\inf_{E\in\calE_{n,:,:}} \frac{f(\calA(E))}{f(\opt(E))}$ (resp. $\sup_{E\in\calE_{n,:,:}} \frac{f(\calA(E))}{f(\opt(E))}$).
We abbreviate the competitive ratio as $\ratio_{n,f}^{\max}(\calA)$ (resp. $\ratio_{n,g}^{\min}(\calA)$).
Further, we use $H_n = \sum_{i\in[n]} \frac{1}{i}$ to denote the $i$th Harmonic number with $H_0 = 0$.

\subsection{Nash Social Welfare}
\begin{lemma}[\nsw\, Competitive Ratio]\label{app:lem:nsw_cr}
    For any number of offline nodes $n\in \N$, \waterfill achieves the minimax optimal competitive ratio for Nash Social Welfare $\nsw(\bell) \define \left(\prod_{i} \bell(i) \right)^{1/n}$ maximization.
    Further, this competitive ratio is given by:
    \begin{align*}
        \ratio_{n,\nsw}^{\max}(\wf)
        = \left(\frac{1}{n!}\right)^{1/n}
    \end{align*}
    In the limit, this competitive ratio satisfies $\lim_{n\to\infty} n\cdot \ratio_{n,\nsw}^{\max}(\wf) = e$.
\end{lemma}

\begin{proof}
    Since \nsw\, is concave and symmetric, \Cref{cor:reg_rand,app:cor:wf_adapt} shows that \wf\, achieves the minimax optimal competitive ratio bound for \nsw\, maximization in comparison to any (even randomized) policy against an oblivious or adaptive adversary.
    It remains to find a closed-form solution for this competitive ratio.
    We generated the a proof of the closed-form solution in conversation with \cite{gemini2025};
    See \Cref{app:sec:disclose} for more details.
    We begin with the lower bound.
    \Cref{cor:reg_bound} shows:
    \begin{align*}
        \ratio_{n,\nsw}^{\max}(\wf)
        &= \inf_{\bell\in\calI_{n,q}}\frac{\nsw(H\bell)}{\nsw(\bell)} \\
        &= \left(\inf_{\bell\in\calI_{n,q}}\prod_{i\in [n]}\left(\frac{\sum_{j\in[i]} \frac{\bell(j)}{n-j+1}}{\bell(i)}\right)\right)^{1/n} \\
        &\ge \left(\prod_{i\in [n]}\left(\frac{1}{n-i+1}\right)\right)^{1/n}
        &&\left(\sum_{j\in[i]} \frac{\bell(j)}{n-j+1} \ge \frac{\bell(i)}{n-i+1} \right)\\
        &= \left(\frac{1}{n!}\right)^{1/n}
    \end{align*}

    To prove an upper bound on the competitive ratio, we use a limiting argument on load vectors.
    To simplify our arguments, we assume $q=1$ since the competitive ratio of \nsw\, is invariant to scaling of $\bell$.
    For a sufficiently small $\epsilon > 0$, define load vector $\bell'\in\calI_{n,q}$ with $\bell'(i) = \epsilon^{n-i}$ for $i\in[n-1]$ and $\bell'(n) = 1-\bell'([n-1])$.
    Computing the competitive ratio as $\epsilon\to 0$ provides an upper bound:
    \begin{align*}
        \ratio_{n,\nsw}^{\max}(\wf)
        &\le \lim_{\epsilon \to 0}\frac{\nsw(H\bell')}{\nsw(\bell')} \\
        &= \left(\lim_{\epsilon \to 0} \prod_{i\in [n]}\left(\frac{\sum_{j\in[i]}\frac{\bell'(j)}{n-j+1}}{\bell'(i)}\right)\right)^{1/n} 
    \end{align*}
    We argue that the limit of each term in the product exists, so we can take the product of each limiting term.
    For $i\in[n-1]$, we have:
    \begin{align*}
        \lim_{\epsilon \to 0} \frac{\left(\sum_{j\in[i]}\frac{\bell'(j)}{n-j+1}\right)}{\bell'(i)}
        &= \sum_{j\in[i]} \lim_{\epsilon\to 0}\frac{\epsilon^{n-j}}{(n-j+1)\epsilon^{n-i}} \\
        &= \frac{1}{n-i+1}
    \end{align*}
    The limit for $i=n$ is:
    \begin{align*}
        \lim_{\epsilon \to 0} \frac{\left(\sum_{j\in[n]}\frac{\bell'(j)}{n-j+1}\right)}{\bell(n)}
        &= 1 + \sum_{j\in[n-1]} \lim_{\epsilon\to 0}\frac{\epsilon^{n-j}}{(n-j+1)(1-\bell'([n-1]))}
        &&\text{Take out $n^{th}$ term} \\
        &= 1
        &&\lim_{\epsilon\to 0} \bell'([n-1]) = 0
    \end{align*}
    Since the limit of each product term exists, the limit of the product is the product of the limits:
    \begin{align*}
        \ratio_{n,\nsw}^{\max}(\wf)
        &= \left(\lim_{\epsilon \to 0} \prod_{i\in [n]}\left(\frac{\sum_{j\in[i]}\frac{\bell'(j)}{n-j+1}}{\bell'(i)}\right)\right)^{1/n} \\
        &= \left(\prod_{i\in [n]}\lim_{\epsilon \to 0}\left(\frac{\sum_{j\in[i]}\frac{\bell'(j)}{n-j+1}}{\bell'(i)}\right)\right)^{1/n} \\
        &= \left(\prod_{i\in[n]} \frac{1}{n-i+1}\right)^{1/n} \\
        &= \left(\frac{1}{n!}\right)^{1/n}
    \end{align*}
    Stirling's approximation for $n!$ proves that this competitive ratio, scaled by $n$, converges to $e$.
\end{proof}

\subsection{Maximin}
\begin{lemma}[\textsc{Maximim}, Competitive Ratio]\label{app:lem:maximin_cr}
    For any number of offline nodes $n\in \N$, \waterfill achieves the minimax optimal competitive ratio for maximizing the Maximin objective $\textsc{Maximin}(\bell) \define \min_{i\in [n]} \bell(i)$.
    Further, this competitive ratio is given by:
    \begin{align*}
        \ratio_{n,\textsc{Maximin}}^{\max}(\wf)
        = \frac{1}{n}
    \end{align*}
\end{lemma}

\begin{proof}
    Since the $\textsc{Maximin}$ objective is symmetric and concave, \waterfill maximizes the competitive ratio in comparison to any alternative (possibly randomized) policy $\calA$ against an oblivious or adaptive adversary due to \Cref{cor:reg_rand,app:cor:wf_adapt}.
    We now find a closed-form solution for the competitive ratio:
    \begin{align*}
        \ratio_{n,\textsc{Maximin}}^{\max}(\wf)
        &= \inf_{\bell\in\calI_{n,:}} \frac{\textsc{Maximin}(H\bell)}{\textsc{Maximin}(\bell)} \\
        &= \frac{\left(\frac{\bell(1)}{n}\right)}{\bell(1)} \\
        &= \frac{1}{n}
    \end{align*}
\end{proof}

\subsection{Maximin}
\begin{lemma}[\textsc{Minimax}, Competitive Ratio]\label{app:lem:minimax_cr}
    For any number of offline nodes $n\in \N$, \waterfill achieves the minimax optimal competitive ratio for minimizing the Minimax objective $\textsc{Minimax}(\bell) \define \min_{i\in [n]} \bell(i)$.
    Further, this competitive ratio is given by:
    \begin{align*}
        \ratio_{n,\textsc{Minimax}}^{\min}(\wf)
        = H_n
    \end{align*}
\end{lemma}

\begin{proof}
    Since the $\textsc{Minimax}$ objective is symmetric and convex, \waterfill minimizes the competitive ratio in comparison to any alternative (possibly randomized) policy $\calA$ against an oblivious or adaptive adversary.
    We now find a closed-form solution for the competitive ratio:
    \begin{align*}
        \ratio_{n,\textsc{Minimax}}^{\min}(\wf)
        &= \inf_{\bell\in\calI_{n,:}} \frac{\textsc{Minimax}(H\bell)}{\textsc{Minimax}(\bell)} \\
        &= \inf_{\bell\in\calI_{n,:}} \frac{\left(\sum_{i\in [n]} \frac{\bell(i)}{n-i+1}\right)}{\bell(n)} \\
        &= H_n
    \end{align*}
    The last equality holds by taking $\bell(1) = \dots = \bell(n) = q$ for any $q>0$.
\end{proof}

\subsection{Fractional Matching}
In the below proofs, $H_n = \sum_{i\in[n]} 1/i$ is the $n$th number in the Harmonic sequence.
We use $H_0 = 0$ for convenience.

\begin{definition}[Fractional Matching Sequence]
    The fractional matching sequence is denoted as $M_1,M_2,\dots\in \R_{\ge 0}$ with $M_n \define \frac{1}{n} \sum_{i=0}^{n} \min(1, H_n - H_i)$. 
\end{definition}

\begin{lemma}[Fractional Matching Competitive Ratio]\label{lem:fm_cr}
    For any number of offline nodes $n\in\N$ and capacity $c\in\R_{>0}$, \waterfill achieves the minimax optimal competitive ratio for fractional matching $\fm_c(\mbf{x}) = \sum_{i\in [n]} \min(c, \mbf{x}(i))$ maximization.
    Further, the competitive ratio of \waterfill for this objective is $\ratio_{n, \fm_c}^{\max} = \min_{k\in[n]} M_k$
\end{lemma}

\begin{proof}
    \Cref{cor:reg_rand} proves that \waterfill yields the minimax optimal competitive ratio for objective $\fm_{c}$ over all allocation policies $\calA$, as this objective is symmetric and concave.
    Without loss of generality, assume that $c=1$, as different capacities are equivalent to scaling inputs: $\fm_c(\mbf{x}) = \fm_1(c\cdot \mbf{x})$.
    We omit the subscript $c$ for clarity for the remainder of the proof.
    
    For an arbitrary hindsight load vector $\bell\in\calI_{n,:}$, we assume that $0 < \bell(1)\le\dots\le \bell(n) = 1$ when trying to maximize the competitive ratio.
    In the case that $\bell(i) \ge 1$ for some $i$, the alternate load vector $\bell'$ with $\bell'(i) \gets \min(1, \bell(i))$ yields $\fm(\bell) = \fm(\bell')$ and $\fm(H\bell) = \fm(H\bell')$ by definition; the competitive ratio is unchanged.
    On the other hand, if $\bell(n) < 1$, then $\bell'\gets \bell/(\bell(n))$ gives $\frac{\fm(\bell)}{\bell(n)} = \fm(\bell')$ and $\frac{\fm(H\bell)}{\bell(n)} \ge \fm(H\bell')$, which only decreases the competitive ratio as desired.

    By allowing zero values $\bell(i) = 0$, we can also assume that $\bell$ takes the form $\bell = (0,\dots,0,\beta,1,\dots,1)$ for some $\beta$.
    For any $\bell$ that does not satisfy the above, notice that the vector $\bell'$ of the desired form with $\bell([n]) = \bell'([n])$ satisfies $\fm(\bell) = \fm(\bell')$ and $\bell \preceq \bell'$.
    Applying \Cref{obs:h_maj} shows $H\bell \preceq H\bell'$ and, since $\fm$ is Schur-concave, this implies $\fm(H\bell) \ge \fm(H\bell')$.
    Thus, the competitive ratio induced by $\bell'$ is no more than that of $\bell$.

    Consider a vector of the form $\bell = (0,\dots,0,\beta,1,\dots,1)$, where $\bell(i) = \beta$ and parameterize the competitive ratio of this vector in terms of $\beta \in [0,1]$.
    We argue that it suffices to consider cases with $\beta\in\{0, 1\}$.
    The competitive ratio is:
    \begin{align*}
        \frac{\fm(H\bell)}{\sum_i \bell(i)}
        = \frac{\min\left(1, \beta / i\right) + \sum_{j>i} \min(1, \beta/i + H_{i-1}) }{\beta + (n - i + 1)}
    \end{align*}
    If $i=n$, then the competitive ratio is $1$ for any $\beta$.
    Otherwise, the ratio is the quotient of a concave function in $\beta$ and a linear function in $\beta$.
    Such functions are always quasi-concave, meaning, the minimum occurs at $\beta\in\{0,1\}$.

    Let $\bell_k$ be the vector with $\bell_k(i) = \1{n-k+1\le i}$.
    Via our previous arguments, the competitive ratio is minimized by such a vector.
    Now, notice that:
    \begin{align*}
        \ratio_{n, \fm}^{\max}(\wf)
        &= \min_{k\in[n]} \frac{\fm(H\bell_k)}{\fm(\bell_k)} \\
        &= \min_{k\in[n]}\frac{\sum_i \min\left(1, \sum_{j} H(i,j)\bell_k(j) \right)}{\sum_i \bell_k(i)} \\
        &= \min_{k\in[n]}\frac{1}{k}\sum_{j}\min\left(1, H_k - H_{j-1}\right) \\
        &= \min_{k\in[n]} M_k \\
    \end{align*}
\end{proof}

\subsection{Separable Concave Functions}
A separable concave function takes the form $f(\mbf{x}) = \sum_i \mathscr{f}(\mbf{x}(i))$ for concave and non-decreasing $\mathscr{f}:\R_{\ge 0} \to \R_{\ge 0}$.
We show that in the class of such functions, the matching objective provides worst-case competitive bounds.

\begin{observation}[Piece-Wise Approximation]\label{obs:concave_decomp}
    Fix a finite set of points $X\in \R_{>0}$ with $r=|X|$ and a function $\mathscr{f}:\R_{\ge 0}\to\R_{\ge 0}$ that is concave and non-decreasing.
    There are constants $\beta_1,\dots,\beta_{r-1}, \gamma\in\R_{\ge 0}$ and capacities $c_1,\dots, c_{r-1}\in \R_{\ge 0}$ such that:
    \begin{align*}
        \mathscr{f}(x) = \gamma x + \sum_{i\in [r-1]} \beta_i \min(c_i, x) , \quad\quad x\in X
    \end{align*}
\end{observation}

\begin{lemma}[Separable Concave Competitive Ratios]\label{lem:sepconcave_cr}
    Fix a separable concave function $f$.
    \waterfill yields minimax optimal competitive ratios when maximizing $f$.
    Further, the competitive ratio of water filling on $f$ satisfies $\ratio_{n,f}^{\max}(\wf) \ge \ratio_{n,\fm}^{\max}(\wf)$.
\end{lemma}

\begin{proof}
    For ease of notation, we write $\rho_n \define \ratio_{n,\fm}^{\max}(\wf)$
    By \Cref{cor:reg_rand}, \waterfill is minimax optimal for separable concave maximization.
    We bound the competitiveness of separable concave functions by showing that fractional matching objectives yield worst-case competitive ratios in this class.
    Fix an arbitrary $\bell\in\calI_{n,:}$ and write $\hat{\bell} \gets H\bell$ for ease of notation.
    By definition, we have $\fm_c(\hat{\bell}) \ge \rho_n\cdot  \fm_c(\bell)$ for any $c>0$.
    The set of values on which we evaluate $\mathscr{f}$ is $X\gets \{\bell(i),\hat{\bell}(i)\}_i$.
    Define $\beta_1,\dots,\beta_{2n-1},\gamma$ and $c_1,\dots,c_{2n-1}$ to be the constants defined in \Cref{obs:concave_decomp} with respect to $\mathscr{f}$.
    We now bound the competitive ratio:
    \begin{align*}
        f(\hat{\bell})
        &= \sum_{i} \mathscr{f}(\hat{\bell}(i)) \\
        &= \sum_{i} \left( \gamma\hat{\bell}(i) + \sum_{j} \beta_j \min(c_j, \hat{\bell}(i)) \right)
        &&\text{\Cref{obs:concave_decomp}} \\
        &= \gamma \sum_i \hat{\bell}(i) + \sum_j \beta_j\left( \sum_i \min(c_j,\hat{\bell}(i)) \right) \\
        &\ge \gamma \sum_i \bell(i) + \rho_n \sum_j \beta_j \left(\min(c_j, \bell(i))\right)
        &&\hat{\bell}([n]) = \bell([n]),\, \fm_{c_j}(\hat{\bell}) \ge \rho_n \cdot \fm_{c_j}(\bell) \\
        &\ge \rho_n \sum_i\left(\gamma \bell(i) + \sum_j \beta_j \min(c_j, \bell(i)) \right) \\
        &= \rho_n \sum_i \mathscr{f}(\bell(i))
        &&\text{\Cref{obs:concave_decomp}} \\
        &= \rho_n f(\bell)
    \end{align*}
    Rearranging the above inequality and taking the supremum on all $\bell\in\calI_{n,:}$ gives:
    \begin{align*}
        \ratio_{n,f}(\wf) = \sup_{\bell\in\calI_{n,:}} \frac{f(H\bell)}{f(\bell)} \ge \rho_n
    \end{align*}
\end{proof}

\section{AI Software Disclosure}\label{app:sec:disclose}
We used Gemini 3 \cite{gemini2025} to assist with an algebraic proof, find closed-form solutions to optimization problems that would characterize competitive ratios for \waterfill, and compress a draft of our related work section.
In addition, we used Refine.ink \cite{refine2026} to identify multiple logic inconsistencies (i.e., incorrect indexing or incorrect variable used), which were corrected by the authors.
All AI-generated content was verified by the authors to the best of their ability.

\end{document}